\documentclass{article}
\usepackage[top=1in, bottom=1in, left=1in, right=1in]{geometry} 
\usepackage{setspace} 
\usepackage{amsfonts}
\usepackage{amsmath}
\usepackage{amsthm} 
\usepackage{amssymb} 
\usepackage{epsfig}
\usepackage{url}
\usepackage{bm}
\usepackage{bbm}
\usepackage{xcolor}

\usepackage{tikz-cd}
\tikzcdset{ampersand replacement=\&}
\usepackage{graphicx}
\usepackage{subcaption}
\usepackage{footmisc}

\usepackage[colorlinks,pagebackref=true]{hyperref}

\DeclareMathOperator{\spn}{span}
\DeclareMathOperator{\Ima}{Im}

\newcommand{\e}{{\rm e}}
\newcommand{\im}{{\rm i}}
\newcommand{\E}{{\mathbb E}}
\newcommand{\Pa}{{\mathbb P}}
\newcommand{\Q}{{\mathbb Q}}
\newcommand{\C}{{\mathbb C}}
\newcommand{\R}{{\mathbb R}}
\newcommand{\M}{{\mathbb M}}
\newcommand{\N}{{\mathbb N}}

\newcommand{\Ecal}{{\mathcal E}}
\newcommand{\Fcal}{{\mathcal F}}

\newcommand{\Hcal}{{\mathcal H}}
\newcommand{\Ical}{{\mathcal I}}

\newcommand{\Ncal}{{\mathcal N}}

\newcommand{\Scal}{{\mathcal S}}

\newtheorem{proposition}{Proposition}[section]
\newtheorem{lemma}[proposition]{Lemma}
\newtheorem{theorem}[proposition]{Theorem}
\newtheorem{definition}[proposition]{Definition}

\newtheorem{remark}[proposition]{Remark}

\newtheorem{exampleemph}[proposition]{Example}   

\begin{document}

\title{Machine learning with kernels for portfolio valuation and risk management\footnote{We thank participants at the Bachelier Finance Society One World Seminar, the SFI Research Days 2020, the CFM-Imperial Quantitative Finance seminar, the David Sprott Distinguished Lecture at Waterloo University, the Workshop on Replication in Life Insurance at Technical University of Munich, the SIAM Conference on Financial Mathematics and Engineering 2019, the 9th International Congress on Industrial and Applied Mathematics, and Kent Daniel, R\"udiger Fahlenbrach, Lucio Fernandez-Arjona, Kay Giesecke, Enkelejd Hashorva, Mike Ludkovski, Markus Pelger, Antoon Pelsser, Simon Scheidegger, Ralf Werner, and two anonymous referees for their comments.}}
\author{ Lotfi Boudabsa\footnote{EPFL. Email: lotfi.boudabsa@epfl.ch} \and Damir Filipovi\'c\footnote{EPFL and Swiss Finance Institute. Email: damir.filipovic@epfl.ch}}
\date{8 April 2021}
\maketitle

\begin{abstract}
We introduce a simulation method for dynamic portfolio valuation and risk management building on machine learning with kernels. We learn the dynamic value process of a portfolio from a finite sample of its cumulative cash flow. The learned value process is given in closed form thanks to a suitable choice of the kernel. We show asymptotic consistency and derive finite sample error bounds under conditions that are suitable for finance applications. Numerical experiments show good results in large dimensions for a moderate training sample size.
\end{abstract}




\section{Introduction}\label{secintro}

Valuation, risk measurement, and hedging form an integral task in portfolio risk management for banks, insurance companies, and other financial institutions. Portfolio risk arises because the values of constituent assets and liabilities of the portfolio change over time in response to changes in the underlying risk factors, e.g., interest rates, equity prices, real-estate prices, foreign exchange rates, credit spreads, etc. The quantification and management of this risk requires a stochastic model of the dynamic portfolio value process.

Most stochastic dynamic models applied in finance can be brought into the following form: an economy with finitely many time periods $t=0,1,\dots,T$, where randomness is generated by some underlying stochastic driver $X=(X_0,\dots,X_T)$. The components $X_t$ are mutually independent, but not necessarily identically distributed, taking values in some measurable spaces $(E_t,\Ecal_t)$. We assume that $X$ is realized on the measurable path space $(E,\Ecal)$, with $E=E_0\times\cdots\times E_T$ and $\Ecal=\Ecal_0\otimes\cdots\otimes\Ecal_T$, such that $X_t(x)=x_t$ for a generic sample point $x=(x_0,\dots,x_T)\in E$. We denote the distribution of $X$ by $\Q(dx) = \Q_0(dx_0)\times\cdots\times \Q_T(dx_T)$, and we assume that $\Q$ represents the risk-neutral pricing measure with respect to some fixed numeraire, such as the money market account. All financial values and cash flows are discounted by this numeraire, if not otherwise stated. The stochastic driver $X$ generates the filtration $\Fcal_t=\Ecal_0\otimes\cdots\otimes\Ecal_t$, which represents the flow of information.

We consider a portfolio whose cumulative cash flow is modeled by some measurable function $f:E\to\R$ such that $f\in L^2_\Q$. Its dynamic value process $V $ is then given by the martingale
\begin{equation}\label{eqcondY}
 V_t=\E_\Q[f(X) \mid\Fcal_t] , \quad t=0,\dots,T.
\end{equation}
Computing $V$ is a notorious challenge, as the conditional expectations \eqref{eqcondY} usually lack analytic solutions. Examples of such portfolios include path-dependent options, structured products, such as barrier reverse convertibles and mortgage-backed securities. Examples also include insurance asset-liability portfolios whose terminal value $f(X)=\sum_{t=1}^T \zeta_t$ is given by accumulating its cash flows $\zeta_t=\zeta_t(X_0,\dots,X_t)\in L^2_\Q$, which are projected in simulations of $X$ and take both financial and insurance risk factors, such as mortality and longevity risks, into account. Similarly, this also includes derivatives trading books held by banks.

This is a very general setup. As an illustrative example, we recall here the multivariate Black--Scholes model, where $X_t$ are i.i.d.\ standard Gaussians on $E_t=\R^d$, for some $d\in\N$. The $d$ nominal stock prices are given by
\begin{equation}\label{SeqBS}
  S_{i,t}=S_{i,t-1}\exp[ \sigma_i^\top X_t \sqrt{{\Delta_t}} + (r- \|\sigma_i\|^2 /2){\Delta_t}],
\end{equation}
for some initial values $S_{i,0}$ and volatility vectors $\sigma_i\in\R^d$, $i=1,\dots,d$, where $r$ is the risk-free rate and $\Delta_t$ denotes the time step size from $t-1$ to $t$ in units of a year. Options on $S$ lack analytic solutions in general. An example is the European max-call option with discounted payoff
\begin{equation}\label{exEMC}
  f(X)=\e^{-r\sum_{t=1}^T \Delta_t} (\max_{i} S_{i,T} -  K)^+,
\end{equation}
for some strike price $K$. We will study this and other examples in more detail below. But note that most of our results apply beyond the Black--Scholes model.

Indeed, we propose a machine learning approach based on kernels with dimension-free error bounds to efficiently compute $V$ in the above general setup. It consists of two steps. First, we approximate $f$ by some function $f_\lambda$ in $L^2_\Q$, where $\lambda> 0$ is a regularization parameter. More specifically, we define $f_\lambda$ as the $\lambda$-regularized projection of $f$ on a suitably chosen reproducing kernel Hilbert space (RKHS) embedded in $L^2_\Q$. Second, we learn $f_\lambda$ from a finite sample $\bm X=(X^{(1)},\dots,X^{(n)})$, drawn from an appropriately chosen equivalent sampling measure $\widetilde\Q\sim\Q$, along with the corresponding function values $f(X^{(i)})$, $i=1,\dots,n$.\footnote{More precisely, $\bm X$ consists of i.i.d.\ $E$-valued random variables $X^{(i)}\sim \widetilde \Q$ defined on the product probability space $(\bm E , \bm\Ecal,  {\bm  Q})$ with $ \bm E = E \otimes E \otimes\cdots $, $ \bm \Ecal =\Ecal \otimes\Ecal\otimes\cdots $, and $  {\bm  Q} = \widetilde \Q \otimes \widetilde \Q\otimes\cdots $.} This yields a sample estimator $f_{\bm X}$ of $f_\lambda$. A suitable choice of the RKHS asserts that the sample estimator
\begin{equation}\label{eqYhat}
  V_{\bm X,t}=\E_\Q[f_{\bm X}(X) \mid\Fcal_t ],\quad t=0,\dots, T,
\end{equation}
of the value process $V$ is given in \emph{closed form}, in the sense that it can be efficiently evaluated at very low computational cost.

How good is this estimator? In view of Doob's maximal inequality, see, e.g., \cite[Corollary II.1.6]{rev_yor_94}, the resulting path-wise maximal $L^2_\Q$-estimation error is bounded by
\begin{equation}\label{doobineq}
\frac{1}{2} \|\max_{t=0,\dots,T} | V_t-V_{\bm X,t}|\|_{2,\Q}\le   \| f -f_{\bm X} \|_{2,\Q}   \le \underbrace{\|f-f_{\lambda}\|_{2,\Q}}_{\text{approximation error}}+\underbrace{\|f_{\bm X} - f_{\lambda}\|_{2,\Q}}_{\text{sample error}}.
\end{equation}
The regularization parameter $\lambda$ can be used to trade off bias for variance and can be chosen optimally through an out of sample validation. More specifically, we show the asymptotic result that the \emph{approximation error} $\|f-f_{\lambda}\|_{2,\Q}$ is minimized as $\lambda\to 0$, and we derive limit theorems and bounds for the \emph{sample error} $\|f_{\bm X} - f_{\lambda}\|_{2,\Q}$. Specifically, we prove asymptotic consistency, $f_{\bm X}\xrightarrow{a.s.} f_\lambda$, and a central limit theorem for $f_{\bm X}- f_\lambda$ in $L^2_\Q$, as the sample size $n\to \infty$. We also derive a finite sample guarantee: for any $\eta\in (0,1]$, there exists a constant $C(\eta)$ such that $ \| f_{\bm X} - f_\lambda\|_{2,\Q}<C(\eta)/\sqrt{n}$ with sampling probability of at least $1-\eta$. All sample error bounds are dimension-free and given by explicit, simple and intuitive expressions in terms of the approximation error $f-f_\lambda$. The smaller the approximation error, the smaller the sample error bounds.

Applications in portfolio risk management are manifold. For a date $t$ we denote by $\Delta V_t  = V_{t}-V_{t-1}$ the gain from holding the portfolio over period $[t-1,t]$. Portfolio risk managers and financial market regulators alike aim to quantify the risk in terms of an $\Fcal_{t-1}$-conditional risk measure, such as value at risk or expected shortfall, evaluated at the loss $-\Delta V_{t}$.\footnote{For the definition of value at risk and expected shortfall (also called conditional value at risk or average value at risk), we refer to \cite[Section~4.4]{foe_sch_04}. See also Section \ref{secexamples} below.} These risk measures refer to the equivalent real-world measure $\Pa\sim\Q$. This calls for a bound on the path-wise maximal $L^1_\Pa$-estimation error, which we readily obtain by combining \eqref{doobineq} with the Cauchy--Schwarz inequality, $\|\max_{t } | V_t-V_{\bm X,t}|\|_{1,\Pa}\le \| \frac{d\Pa}{d\Q}\|_{2,\Q}  \|\max_{t } ( V_t-V_{\bm X,t})\|_{2,\Q}$. Indeed, this provides a bound on the estimation error for risk measures which are continuous with respect to the $L^1_\Pa$-norm, such as value at risk (under mild technical conditions) and expected shortfall, see e.g. \cite[Section 6]{cam_fil_17}.

Another important task of portfolio risk management is hedging. The risk exposure from holding the portfolio over period $[t-1,t]$ can be mitigated by replicating its value process through dynamic trading in liquid financial instruments. Let $G$ be a vector of $L^2_\Q$-martingales that models the discounted value processes of tradeable financial instruments. We find the $\Q$-variance optimal hedging strategy by projecting $\Delta V_{t}$ on the gains of the financial instruments $\Delta G_{t}=G_t-G_{t-1}$, that is, by minimizing $\E_\Q[(\psi^\top_{t-1} \Delta G_{t} - \Delta V_{t})^2\mid\Fcal_{t-1}]$ over all $\Fcal_{t-1}$-measurable vectors $\psi_{t-1}$. The solution is given by 
\begin{equation}\label{psieqhedge}
 \psi_{t-1} = \E_\Q[ \Delta G_{t}  \Delta G_{t}^\top\mid\Fcal_{t-1}]^{ -1} \,\E_\Q[ \Delta G_{t} \Delta V_{t}\mid\Fcal_{t-1}],
\end{equation}
see, e.g., \cite[Chapter 10]{foe_sch_04}.

Summarizing, for either of these portfolio risk management tasks, we have to compute the dynamic value process $V$. This is a computational challenge, as the conditional expectations \eqref{eqcondY} usually lack analytic solutions. What's more, in real-life applications in the portfolio management industry, the point-wise evaluation of $f$ is costly, because it queries from various constituent sub-portfolios, which in practice are often not implemented on one integrated platform. For illustration, a technical report of the German Actuarial Society \cite{dav_15} reports as typical sample size in practice of $n=1000$ to $5000$. Similarly, \cite{hug_sav_20} state that learning effectively from small datasets is critical in the context of regulations of complex derivatives trading books held by banks. In practice, this amounts to sample sizes of $n=1000$ to $32000$, as reported in \cite{hug_sav_19}. Facing a limited computing budget calls for an efficient method to approximate and learn the value process $V$ from a (small) finite sample and in such a way that the sample estimator is given in closed form, such as in \eqref{eqYhat}. This is exactly what our paper provides. As for the dimension of the path space $E$, we target a range of $dT\le 36$, which can be considered high-dimensional in practical terms, as discussed in \cite[Section~5]{fer_fil_2020}.

Our paper builds on the vast literature on machine learning with kernels, which has its roots in the early works of James Mercer \cite{mer_1909} and Stefan Bergman \cite{ber_1922} who studied integral operators related to kernels. The basic theory of RKHS's was developed in the seminal paper \cite{aro_50}. Kernels were rediscovered by the machine learning community in the 1990s and utilized for nonlinear classification \cite{bos_guy_vap_92} and nonlinear PCA \cite{sch_smo_mue_98}. This boosted an extensive research activity on kernel based learning. \cite{sun_05} and \cite{ste_sco_12} provide a systematic functional analysis of kernels on general (i.e., non-compact) domains, \cite{dev_ros_cap_05} connect the theories of statistical learning and ill-posed problems via Tikhonov regularization, \cite{ros_bel_dev_10} study convergence of integral operators using a concentration inequality for Hilbert space-valued random variables. Our sample estimators are based on kernel ridge regression, which is discussed in detail in, e.g., \cite{cuck_smale_2001, wu_et_al_2006}. We add to this literature by developing a tailor-made framework of kernel ridge regression for dynamic portfolio valuation and risk management. To the best of our knowledge, related results in the machine learning literature are derived under stringent assumptions on either $f$ (e.g., bounded and smooth in \cite{ras_sam_2017,cap_dev_2007}) or $E$ (e.g., compact in \cite{lin_et_al_2018}), which do not hold in applications in finance. This is evident from the simple example~\eqref{exEMC} above. Moreover, we exploit the celebrated kernel representer theorem for obtaining closed form estimators of the value process. Modern introductory texts to machine learning with kernels include \cite{sch_smo_02,bis_06, cuc_zho_07, hof_sch_smo_08, ste_chr_08, pau_rag_16}. For the convenience of the reader we recall the essentials of Hilbert spaces, and RKHS's in particular, in the appendix.

The literature related to portfolio risk measurement includes  \cite{bro_du_moa_15} who introduce a regression-based nested Monte Carlo simulation method for the estimation of the unconditional expectation of a Lipschitz continuous function $f(L)$ of the 1-year loss $L=-\Delta V_{1}$. They also provide a comprehensive literature overview of nested simulation problems, including \cite{gor_jun_10} who improve the speed of the convergence of the standard nested simulation method using the jackknife method. Our method is different as it learns the entire value process $V$ in one go, as opposed to any method relying on nested Monte Carlo simulation, which estimates $V_t$ for one fixed $t$ at a time.

Specific literature on insurance liability portfolio replication includes \cite{nat_wer_14,pel_sch_16,cam_fil_18}. Learning functions in the context of uncertainty quantification includes \cite{coh_mig_17}. These papers have in common that they project $f$ on a finite set of basis functions. As such they are contained in our unified framework as special cases of finite-dimensional RKHS's with $\lambda=0$. An infinite-dimensional approach is given in \cite{ris_lud_16,ris_lud_18}, who learn the value process using Gaussian process regression.

Here and throughout we use the following conventions and notation. For a probability space $(E,\Ecal,\Q)$, for $p\in [1,\infty]$, and for measurable functions $f,g:E\to\R$, we denote
\[ \|f\|_{p,\Q}=\begin{cases}
  (\int_E |f(x)|^p\Q(dx))^{1/p},& p<\infty,\\
  \inf\{c \ge 0\mid \text{$|f|\le c$ $\Q$-a.s.}\}, & p=\infty,
\end{cases}\]
and $\langle f,g\rangle_{\Q}=\int_E f(x)g(x)\Q(dx)$, whenever $\|fg\|_{1,\Q}<\infty$. We denote by $L^p_\Q$ the space of \emph{$\Q$-equivalence classes} of measurable functions $f:E\to\R$ with $\|f\|_{p,\Q}<\infty$. If not otherwise stated, we will use the same symbol, e.g., $f$, for a function and its equivalence class. Every $L^p_\Q$ is a separable Banach space with norm $\|\cdot\|_{p,\Q}$, and $L^2_\Q$ is a separable Hilbert space with inner product $\langle \cdot,\cdot\rangle_{\Q}$. We denote by $\|y\|=\sqrt{y^\top y}$ the Euclidian norm of a coordinate vector $y$. Various operator norms on Hilbert spaces are introduced in Section~\ref{seccopoH}.

The remainder of the paper is as follows. Section~\ref{secapprox} discusses the kernel-based approximation of $f$. Section~\ref{secFSE} contains the sample estimation and error bounds. Section~\ref{seccompnew} provides computational formulas for the sample estimator and gives the estimated value process in closed form. Section~\ref{secTK} presents a large class of tractable kernels. Section~\ref{secexamples} provides numerical examples for the valuation of path-dependent, exotic options in the Black--Scholes model. In particular, we compute value at risk and expected shortfall of long and short positions, and we sketch the implementation of the $\Q$-variance optimal hedging. Section~\ref{secconc} concludes. Section~\ref{secfactsH} recalls some facts about Hilbert spaces, including the essentials of RKHS's, compact operators, and random variables in Hilbert spaces. Section~\ref{secproofs} contains all proofs from the main text. Sections~\ref{secdimL2n} and \ref{secdimH} are auxiliary and discuss in more detail the cases where the target space and the RKHS are finite dimensional, respectively. Section \ref{appregnow} briefly discusses the regress-now approach introduced in \cite{gla_yu_04} and compares it with our approach.

\section{Approximation}\label{secapprox}

As in Section~\ref{secintro}, we let $f\in L^2_\Q$ be a given function, which models the payoff, or cumulative cash flow, of a portfolio. We approximate and learn $f$ through the choice of an appropriate hypothesis space $\Hcal$ embedded in $L^2_\Q$. Thereto, we choose a \emph{kernel} $k$ on $E$. That is, a function $k:E\times E\to\R$ such that, for any finite selection of points $x_1,\dots,x_n\in E$, the $n\times n$-matrix with entries $k(x_i,x_j)$ is symmetric and positive semidefinite. By Moore's theorem \cite[Theorem 2.14 and Proposition 2.3]{pau_rag_16}, there exists a unique \emph{reproducing kernel Hilbert space (RKHS)} $\Hcal$ with kernel $k$. That is, a Hilbert space $\Hcal$ consisting of functions $h:E\to\R$ such that $k(\cdot,x)$ is in $\Hcal$ and acts as pointwise evaluation, $\langle h,k(\cdot,x) \rangle_\Hcal=h(x)$, for all $x\in E$. We collect some basic properties of RKHS in Section~\ref{secfactsH}.

Throughout, we assume that $k:E\times E\to\R$ is measurable and $\Hcal$ is separable.\footnote{Sufficient conditions for separability of an RKHS are given in Lemma~\ref{lemsepHnew}.} Then $\Hcal$ consists of measurable functions, see \cite[Theorem~90]{ber_tho_2004}.  We also assume that $\kappa(x)=\sqrt{k(x,x)}=\|k(\cdot,x)\|_\Hcal$ is square-integrable, \begin{equation}\label{ass0}
  \|\kappa\|_{2,\Q}<\infty.
\end{equation}
From the elementary bound
\begin{equation}\label{eqfundamentalpnew}
|h(x)|\le\kappa(x)\|h\|_{\Hcal},\quad x\in E,
\end{equation}
we then infer that the linear operator $J:\Hcal\to L^2_\Q$ that maps $h\in\Hcal$ to its $\Q$-equivalence class is well-defined and Hilbert--Schmidt with norm $\|J\|_2=\|\kappa\|_{2,\Q}$, see \cite[Lemma~2.3]{ste_sco_12}.\footnote{By \eqref{eqfundamentalpnew}, we have that $J:\Hcal\to L^p_\Q$ is a bounded operator with $\|J\|\le   \|\kappa\|_{p,\Q}$, for any $p\le\infty$ such that $\|\kappa\|_{p,\Q}<\infty$.\label{fn1}} It is well known, see \cite[Lemma~2.2]{ste_sco_12}, that the adjoint operator $J^\ast:L^2_\Q\to\Hcal$ satisfies
  \begin{equation}\label{eqJastgnew}
   J^\ast g =\int_E k(\cdot,x) g(x)\Q(dx),\quad  g\in L^2_\Q .
  \end{equation}

We can now approximate $f$ in $L^2_\Q$ by the solution $h=f_\lambda\in\Hcal$ to the regularized projection problem
\begin{equation}\label{KRR}
  \min_{h\in\Hcal} (\|f-h\|^2_{2,\Q} + \lambda\|h\|^2_{\Hcal}),
\end{equation}
for some regularization parameter $\lambda> 0$. There are two arguments for adding the penalization term $\lambda\|h\|^2_{\Hcal}$ in the objective function~\eqref{KRR}. First, we avoid overfitting when $\Hcal$ is relatively ``large'' compared to $L^2_\Q$, in the sense that $\overline{\Ima J} = L^2_\Q$, which happens in particular when $\dim (L^2_\Q)<\infty$, as described in Section~\ref{secdimL2n} and the sample estimation below. Second, problem~\eqref{KRR} has always a unique solution $h=f_\lambda\in\Hcal$ and it is given by
\begin{equation}\label{eqKRR}
  f_\lambda = (J^\ast J+\lambda)^{-1} J^\ast f,
\end{equation}
see \cite[Theorem 5.1]{eng_et_al_1996}. It readily follows from \eqref{eqJastgnew} and \eqref{eqKRR} that $f_\lambda$ can be represented as
\begin{equation}\label{hsg}
 f_\lambda =J^\ast g_\lambda =\int_E k(\cdot,x) g_\lambda(x)\Q(dx)
\end{equation}
where
\begin{equation}\label{eqKRRalt}
  g_\lambda = (JJ^\ast+\lambda)^{-1}  f.
\end{equation}

Equation~\eqref{hsg} is known as \emph{representer theorem}, see, e.g., \cite[Section 8.6]{pau_rag_16}. It yields an important lemma for applications in finance, as we shall see next. For the definition of kernel embeddings of distributions we refer to \cite{sri_etal_10}.
\begin{definition}
 We call the kernel $k$ \emph{tractable} if the conditional kernel embeddings $M_t(y)=\E_\Q[k(X,y)\mid\Fcal_t]$ are given in closed form, for all $y\in E$, $t=0,\dots,T$.
\end{definition}

\begin{lemma}\label{lemclosedform}
 Assume that $k$ is tractable and let $g_\lambda$ be given by \eqref{eqKRRalt}. Then
\begin{equation}\label{intreprescf}
  \E_\Q[f_\lambda(X)\mid\Fcal_t] =  \int_E M_t(y) g_\lambda(y) \Q(dy)
\end{equation}
is given in closed form, subject to $\Q$-integration, for all $t=0,\dots,T$.\footnote{The integral in \eqref{intreprescf} boils down to a finite sum in the sample estimation of $f_\lambda$ below, see Lemma~\ref{lemcompnewX}.}
\end{lemma}

We now discuss the limit $\lambda\to 0$. Thereto, we denote by $f_0\in \overline{\Ima J}$ the orthogonal projection of $f$ onto $\overline{\Ima J}$ in $L^2_\Q$. By orthogonality of $f - f_0$ and $f_0 - f_\lambda$ in $L^2_\Q$, we can decompose the squared \emph{approximation error}
\[ \| f - f_\lambda\|^2_{2,\Q} = \| f - f_0\|^2_{2,\Q}+\| f_0 - f_\lambda\|^2_{2,\Q}\]
into the sum of the squared \emph{projection error} $\| f - f_0\|_{2,\Q}$ and the squared \emph{regularization error} $\| f_0 - f_\lambda\|_{2,\Q}$. The next result is well known and shows that the regularization error converges to zero as $\lambda\to 0$, albeit the convergence may be slow, see \cite[Proposition 4]{dev_ros_cap_05}.\footnote{In fact, $\{J(J^\ast J+\lambda)^{-1} J^\ast\mid\lambda >0\}$ is a bounded family of operators on $L^2_\Q$, with $\|J(J^\ast J+\lambda)^{-1} J^\ast  \|\le 1$ by Section~\ref{ssecHS}, which converges weakly to the projection operator onto $\overline{\Ima J}$, $f_\lambda\to f_0$ as $\lambda\to 0$, but not so in operator norm in general. }

\begin{lemma}\label{lemconv}
 $\| f_0 - f_\lambda\|_{2,\Q}\to 0$ as $\lambda\to 0$.
\end{lemma}

In view of Lemma~\ref{lemconv}, the following property of $k$ is desirable because it implies a zero projection error, $f_0=f$, so that the approximation error converges to zero as $\lambda\to 0$.\footnote{Universal kernels have been introduced by \cite{ste_02, mic_xu_zha_06}. See also \cite{sri_fuk_lan_10}.}

\begin{definition}\label{defL2UK}
The kernel $k$ is called \emph{$L^2_\Q$-universal} if $\overline{\Ima J}=L^2_\Q$.
\end{definition}

We discuss the special cases of a finite-dimensional target space $L^2_\Q$ and a finite-dimensional RKHS $\Hcal$ in more detail in Sections~\ref{secdimL2n} and \ref{secdimH}.

A standard assumption in the machine learning literature is that $f_0\in\Ima J$, which holds if and only if problem \eqref{KRR} has a solution for $\lambda=0$. Under this regularity assumption, one can derive rates of convergence in Lemma~\ref{lemconv}, see, e.g., \cite{cap_dev_2007}. However, note that this assumption is quite restrictive and difficult to verify in practice, unless the RKHS $\Hcal$ is finite dimensional.\footnote{As $J:\Hcal\to  L^2_\Q$ is a compact operator, by the open mapping theorem, we have that $\overline{\Ima J}=\Ima J$ if and only if $\dim(\Ima J)<\infty$. In this case, obviously, $f_0\in \Ima J$.} In this paper, we thus abstain from this assumption. We henceforth acknowledge the approximation error for a given $\lambda>0$, which thanks to Lemma~\ref{lemconv} and Definition~\ref{defL2UK} can be assumed arbitrarily small, and focus on the sample error in the sequel.

\section{Sample estimation}\label{secFSE}

We next learn the approximation $f_\lambda$ from a finite sample. The previous machine learning literature has derived sample error bounds under regularity and boundedness assumptions on $f$ and $k$ that do not hold for finance applications in general. We thus add to the literature with the following setup.

We therefore transform $f$ and $k$ into bounded functions and compensate for this transformation by sampling under some alternative measure, if necessary. Specifically, we fix an equivalent sampling measure $\widetilde\Q\sim\Q$ with Radon--Nikodym derivative $w=d{\widetilde\Q}/d\Q$, and we define the measurable function ${\widetilde f}=f/\sqrt{w}$ and measurable kernel $\widetilde k(x,y)=k(x,y)/\sqrt{w(x) w(y)}$. We assume that $w$ is chosen such that
\begin{equation}
    \| \widetilde f\|_{\infty,\Q}<\infty\label{newasstildef}
\end{equation}
and
\begin{equation}
   \| \widetilde\kappa\|_{\infty,\Q}<\infty \label{newasstildek}
\end{equation}
where we define $\widetilde\kappa(x)=\sqrt{\widetilde k(x,x)}=\kappa(x)/\sqrt{w(x)}$.\footnote{As in footnote \ref{fn1}, in view of \eqref{eqfundamentalpnew} and \eqref{newasstildek}, we necessarily have $\|\kappa\|_{p,\Q}\le \|\sqrt{w}\|_{p,\Q}\| \widetilde\kappa\|_{\infty,\Q}<\infty$, for any $p\le\infty$ such that $\|\sqrt{w}\|_{p,\Q}<\infty$. The last obviously holds for $p=2$.}

We denote by $\widetilde\Hcal$ the RKHS corresponding to $\widetilde k$. It is readily seen that the linear operator $U:L^2_{\widetilde\Q}\to L^2_{\Q}$ given by $U   g =\sqrt{w}   g$ is an isometry, with $U^{-1} g=U^\ast g=g/\sqrt{w}$. Hence $\|\widetilde f\|_{2,\widetilde \Q}=\|f\|_{2,\Q}$ and $\|\widetilde\kappa\|_{2,\widetilde\Q}=\|\kappa\|_{2,\Q}$. Moreover, from \cite[Proposition 5.20]{pau_rag_16} we infer that the linear operator $T:\widetilde\Hcal\to \Hcal$ given by $T   h = \sqrt{w}   h$ is an isometry, with $T^{-1}h=T^\ast h = h/\sqrt{w}$. As a consequence, $\widetilde\Hcal$ is separable and the following diagram commutes, in the sense that $\widetilde J = U^{-1}  J  T$,
\begin{equation}\label{diagHtilde}
  \begin{tikzcd}
\widetilde\Hcal \arrow{r}{\widetilde J} \arrow[d,"\times\sqrt{w}"]
\& L^2_{\widetilde\Q}  \\
\Hcal \arrow{r}{J}
\& L^2_{\Q}\arrow[u, "\times\frac{1}{\sqrt{w}}"']
\end{tikzcd}
\end{equation}
where $\widetilde J:\widetilde \Hcal\to L^2_{\widetilde\Q}$ denotes the linear operator that maps $  h\in \widetilde \Hcal$ to its $\Q$-equivalence class. As a consequence, all results of Section~\ref{secapprox} can be lifted and literally apply to $\widetilde \Q$, $\widetilde k$, $\widetilde \Hcal$, $\widetilde J$, $\widetilde f$ in lieu of $\Q$, $k$, $\Hcal$, $J$, $f$. In particular, we obtain the approximation $\widetilde f_\lambda$ of $\widetilde f$ in $L^2_{\widetilde\Q}$, and we have $f_\lambda=\sqrt{w} \widetilde f_\lambda$. Note also that $\widetilde k$ is $L^2_{\widetilde \Q}$-universal if and only if $k$ is $L^2_\Q$-universal.

We now let $n\in\N$ and $\bm X=(X^{(1)},\dots,X^{(n)})$ be a sample of i.i.d.\ $E$-valued random variables with $X^{(i)}\sim\widetilde\Q$. Without loss of generality we assume that the random variables $X^{(i)}$ are defined on the product measurable space $\bm E=E\times E\times\cdots$ and $\bm \Ecal= \Ecal\otimes\Ecal\otimes\cdots$, endowed with the product probability measure $\bm Q=\widetilde\Q \otimes\widetilde\Q \otimes\cdots$.

We define the empirical measure $\widetilde\Q_{\bm X}=\frac{1}{n}\sum_{i=1}^n \delta_{X^{(i)}}$ on $E$. Then, again, all results of Section~\ref{secapprox} apply sample-wise for $\widetilde\Q_{\bm X}$ in lieu of $\widetilde\Q$. We denote by $\widetilde J_{\bm X}:\widetilde\Hcal\to L^2_{\widetilde\Q_{\bm X}}$ and $\widetilde f_{\bm X}= (\widetilde J_{\bm X}^\ast \widetilde J_{\bm X} +\lambda)^{-1} \widetilde J_{\bm X}^\ast \widetilde f$ the sample analogues of $\widetilde J:\widetilde\Hcal\to L^2_{\widetilde\Q}$ and $\widetilde f_\lambda$, respectively.\footnote{As above, for any function $h:E\to\R$, we will write $h$ for its $\widetilde\Q_{\bm X}$-equivalence class. } Consistently with \eqref{diagHtilde}, we eventually \emph{define} the sample estimator of $f_\lambda$ by
\begin{equation}\label{eqFXdef}
 f_{\bm X} =   \sqrt{w}\widetilde f_{\bm X}.
\end{equation}

Our first main result is a pair of limit theorems, which shows consistency of the estimator \eqref{eqFXdef} seen as a function in $\Hcal$. For the notion of a Gaussian measure $\Ncal(m,Q)$ with mean $m$ and covariance operator $Q$ on a Hilbert space, we refer to Section~\ref{appLTH}. We denote the variance of $g\in L^2_{\widetilde\Q}$ by $\mathbb{V}_{\widetilde\Q}[g]=\|g\|_{2,{\widetilde\Q}}^2 - \langle g,1\rangle_{\widetilde\Q}^2 $.

 \begin{theorem}\label{thmLT}
\begin{enumerate}
\item Law of large numbers in $\Hcal$: $f_{\bm X}\xrightarrow{a.s.} f_\lambda$ as $n\to \infty$.

\item Central limit theorem in $\Hcal$: $\sqrt{n}( f_{\bm X}-f_\lambda ) \xrightarrow{d}  \Ncal(0, Q)$ as $n\to \infty$, where $Q:\Hcal\to \Hcal$ is the nonnegative, self-adjoint trace-class operator given by $\langle Q h ,h\rangle_\Hcal=\mathbb{V}_{\widetilde\Q}[(1/w)(f  - f_{\lambda} ) (J^\ast J+\lambda)^{-1} h] $, for $h\in \Hcal$.
\end{enumerate}
\end{theorem}

An immediate consequence of Theorem~\ref{thmLT} is the following weak central limit theorem, which holds for any $h\in \Hcal$,
\begin{equation}\label{weakCLT}
  \sqrt{n}  \langle f_{\bm X}-f_\lambda, h \rangle_\Hcal \xrightarrow{d}  \Ncal(0, \langle Q h ,h\rangle_\Hcal) \quad\text{ as $n\to \infty$.}
\end{equation}

\begin{remark}\label{remCLT}
From Theorem \ref{thmLT} and the continuous mapping theorem we immediately obtain the corresponding law of large numbers and central limit theorem in $L^2_\Q$. The latter reads $\sqrt{n}( f_{\bm X}-f_\lambda ) \xrightarrow{d}  \Ncal(0, J Q J^\ast)$ as $n\to \infty$, where $J Q J^\ast:L^2_\Q\to L^2_\Q$ is the nonnegative, self-adjoint trace-class operator given by $\langle  J Q J^\ast g ,g\rangle_\Q=\mathbb{V}_{\widetilde\Q}[(1/w)(f-f_\lambda)(J^\ast J+\lambda)^{-1} J^\ast g]$, for $g\in L^2_\Q$. The weak central limit theorem~\eqref{weakCLT} reads $\sqrt{n}  \langle f_{\bm X}-f_\lambda, g \rangle_\Q \xrightarrow{d}  \Ncal(0, \langle J Q J^\ast g ,g\rangle_\Q)$ as $n\to \infty$.
\end{remark}

\begin{remark}\label{remthmLT}
Theorem \ref{thmLT} actually holds under weaker assumptions than \eqref{newasstildef}--\eqref{newasstildek}, namely $\|\widetilde f \widetilde\kappa \|_{2,\widetilde\Q}<\infty$ and $\|\widetilde\kappa\|_{4,\widetilde\Q}<\infty$. This is evident from its proof, see \eqref{xi_expec} and \eqref{decJaJJaJ3}.
\end{remark}

Our second main result gives finite sample guarantees for the estimator \eqref{eqFXdef}.

\begin{theorem}\label{thmcinew}
For any $\eta\in (0,1]$, we have
\begin{equation}\label{eqCInew}
  \|f_{\bm X} - f_\lambda\|_\Hcal  <  \frac{2\sqrt{2 \log(2/\eta)}\|(1/w)(f-f_\lambda)\kappa\|_{\infty,\Q}}{\lambda\sqrt{n}}
\end{equation}
with sampling probability $\bm Q$ of at least $1-\eta$.
\end{theorem}

\begin{remark}\label{remdimfree}
Note that the bound in Theorem \ref{thmcinew} is dimension-free in the sense that, while the constants may depend on the dimension of $E$, the convergence rate in $n$ does not. From the proof, we see that this is a direct consequence of the Hoeffding inequality \eqref{hoeffdingeq}, which is dimension-free. We also observe that the closer the approximation $f_\lambda$ to $f$, the smaller the finite sample error bounds.
\end{remark}

As for the choice of the sampling measure ${\widetilde\Q}$ that satisfies conditions~\eqref{newasstildef} and \eqref{newasstildek}, there is an optimal one that yields a minimal $L^\infty$-norm of the kernel in the following sense.
\begin{lemma}\label{lemoptwnew}
For any sampling measure ${\widetilde\Q}\sim\Q$, we have $\|\widetilde\kappa\|_{\infty,{\Q}}\ge \|\kappa\|_{2,\Q}$, with equality if and only if $\kappa> 0$ and
\begin{equation}\label{natwasnew}
  w = \frac{\kappa^2}{\|\kappa\|_{2,\Q}^2},\quad \text{$\Q$-a.s.}
\end{equation}
In this case, $\widetilde \kappa=\|\kappa\|_{2,\Q}$ is constant $\Q$-a.s.
\end{lemma}

With the choice \eqref{natwasnew} we obtain that $\|\widetilde\kappa\|_{\infty,\Q}=\|\kappa\|_{2,\Q}$, which asserts condition \eqref{newasstildek}. As for condition~\eqref{newasstildef}, in conjunction with the choice \eqref{natwasnew}, we can always choose the original kernel $k$ such that $\|f/\kappa\|_{\infty,\Q}<\infty$, which then implies \eqref{newasstildef}.

Besides the above considerations, for practical matters, it is convenient to choose the sampling measure ${\widetilde\Q}\sim\Q$ such that
\begin{equation}\label{sampnu}
 \text{sampling from ${\widetilde\Q}$ is feasible.}
\end{equation}

Finite sample guarantees similar to \eqref{eqCInew} have been derived in the machine learning literature, e.g., \cite{cuck_smale_2001, cuc_sma_2002, ste_sco_2005, wu_et_al_2006, wu_zho_2006, cap_dev_2007, bau_et_al_2007, sma_zho_2007, wu_et_al_2007,  ras_sam_2017, lin_et_al_2018}, but under more stringent assumptions than ours. For instance, \cite{ras_sam_2017} assume that $f_0\in\Ima J$, which does not hold in our examples in Section~\ref{secexamples} below. Indeed, the Gaussian-exponentiated kernel is $L^2_\Q$-universal, see Lemma~\ref{lemGEKuniversalnew}, and hence $f_0=f$ which is not in $\Ima J$ for any of the payoff functions $f$. For another instance, \cite{lin_et_al_2018} assume that $E$ is compact, which again does not hold in our examples. Such assumptions are standard in the above literature, where most papers aim to determine optimal learning rates for the total error $\|f_{\bm X} - f_0\|_{2,\Q}$. These are statements of the form $\bm Q[\|f_{\bm X} - f_0\|_{2,\Q}  < c(\eta,n)]\ge 1-\eta$ for all $n\ge n_0(\eta)$, for $\eta\in (0,1]$. The best learning rate obtained so far is $c(\eta,n) = O(n^{-1/2})$, which is consistent with \eqref{eqCInew}. However, we believe that separating the sample error bound \eqref{eqCInew}, for a fixed $\lambda>0$, from the approximation error leads to higher transparency of the arguments and the flexibility of our framework to adhere to financial applications. As a consequence, we state Theorem \ref{thmcinew} under the minimal assumptions in order to cover financial examples. Indeed, finite sample guarantees stated in the machine learning literature hold under more stringent assumptions and would require from the user in finance to carefully inspect these proofs, to check whether they also apply under our weaker assumptions. This is arguably a cumbersome task. For the sake of the reader, we therefore give a self-contained short proof of Theorem \ref{thmcinew}. The arguments are in the spirit of the proof of \cite[Theorem~3.1]{ras_sam_2017}, which again builds on \cite[Theorem~2]{dev_ros_cap_05} and \cite[Theorem~4]{cap_dev_2007}, and which, however, is stated under the aforementioned stringent assumptions. Interestingly, to the best of our knowledge, Theorem \ref{thmLT} is not available in the literature in this form.


\section{Computation}\label{seccompnew}

We show how to compute $ f_{\bm X}$. We also derive the sample analogue of Lemma~\ref{lemclosedform}, which gives the estimated value process $ V_{\bm X}$ in \eqref{eqYhat} in closed form. We explicitly take into account that sample points may overlap.

We start by noting that $\bar n=\dim L^2_{\widetilde\Q_{\bm X}}\le n$, with equality if and only if
\begin{equation}\label{XiXjneq}
 X^{(i)}\neq X^{(j)}\quad \text{for all $i\neq j$.}
\end{equation}
Therefore, we let $\bar X^{(1)},\dots,\bar X^{(\bar n)}$ be the distinct points in $E$ such that $\{\bar X^{(1)},\dots,\bar X^{(\bar n)}\}=\{X^{(1)},\dots,X^{(n)}\}$.\footnote{This sorting step adds computational cost. In Section~\ref{sseccompwos} we show how to compute $ f_{\bm X}$ without sorting.} Define the index sets $I_j =\{ i\mid X^{(i)}=\bar X^{(j)}\}$, $j=1,\dots,\bar n$. We consider the orthogonal basis $\{\psi_1,\dots,\psi_{\bar n}\}$ of $L^2_{\widetilde\Q_{\bm X}}$ given by $\psi_i(\bar X^{(j)})=  |I_i|^{-1/2} \delta_{ij}$, so that $\langle\psi_i,\psi_j\rangle_{\widetilde\Q_{\bm X}}=\frac{1}{n}\delta_{ij}$, for $1\le i,j\le \bar n$. The coordinate vector representation of any $g\in L^2_{\widetilde\Q_{\bm X}}$ accordingly is given by
\begin{equation}\label{eqcoordtilde}
  {\bm g} = (|I_1|^{1/2} g(\bar X^{(1)}),\dots,|I_{\bar n}|^{1/2} g(\bar X^{(\bar n)}))^\top.
\end{equation}
We define the positive semidefinite $\bar n\times \bar n$-matrix ${\bm K}$ by ${\bm K}_{ij}= |I_i|^{1/2}\widetilde k(\bar X^{(i)},\bar X^{(j)}) |I_j|^{1/2}$, for $1\le i,j\le \bar n$. From \eqref{eqJxJxa} we see that $\frac{1}{n}{\bm K}$ is the matrix representation of $\widetilde J_{\bm X}\widetilde J_{\bm X}^\ast: L^2_{\widetilde\Q_{\bm X}}\to L^2_{\widetilde\Q_{\bm X}}$. We thus arrive at the following lemma, which shows how to compute $ f_{\bm X}$ and $ V_{\bm X}$ in terms of ${\bm K}$ and ${\bm f}=(|I_1|^{1/2} \widetilde f(\bar X^{(1)}),\dots,|I_{\bar n}|^{1/2} \widetilde f(\bar X^{(\bar n)}))^\top $, the coordinates of $\widetilde f$ in $L^2_{\widetilde\Q_{\bm X}}$ according to \eqref{eqcoordtilde}.

\begin{lemma}\label{lemcompnewX}
The unique solution ${\bm g}\in\R^{\bar n}$ to
\begin{equation}\label{KLS'}
  \textstyle(\frac{1}{n}{\bm K}+\lambda) {\bm g} = {\bm f},
\end{equation}
gives $f_{\bm X}=\frac{1}{n}\sum_{j=1}^{\bar n} k(\cdot,\bar X^{(j)})\frac{|I_j|^{1/2} {\bm g}_j}{\sqrt{w(\bar X^{(j)})}}$. If, moreover, the kernel $k$ is tractable then
\begin{equation}\label{hatVcf1}
  V_{\bm X,t}=
  \frac{1}{n}\sum_{j=1}^{\bar n} M_t(\bar X^{(j)}) \frac{|I_j|^{1/2} {\bm g}_j}{\sqrt{w(\bar X^{(j)})}} ,\quad t=0,\dots,T,
\end{equation}
is given in closed form.
\end{lemma}

\begin{remark}
Computing the $\bar n\times \bar n$-matrix ${\bm K}$ is infeasible when $\bar n$ is significantly greater than $10^5$ both in
terms of memory and computation, see \cite{mai_ver_18}. In this case, one could consider a low-rank approximation of the kernel of the form $\widetilde k(x,y)\approx \widetilde\phi(x)^\top\widetilde\phi(y)$ for some feature map $\widetilde\phi:E\to\R^m$. This brings us to the finite-dimensional case discussed in Lemma~\ref{lemcompnewFDnew} below. There has recently been a lot of research on such low-rank approximations of kernels. E.g., \cite{dai_etal_14,lu_etal_16} use a probabilistic representation of the kernel as in Lemma~\ref{thmmercernew}\ref{thmmercernew2}, where they approximate $\M$, and thus $k$, by the empirical measure induced by a finite sample $\omega_1,\dots,\omega_m\in\Omega$ drawn from $\M$.
\end{remark}

\section{Tractable kernels}\label{secTK}

As we have seen, the above kernel method can be efficiently applied for approximating $V$ if the chosen kernel is tractable for a given random driver. Luckily there are many such kernels $k$ and distributions $\Q$, as we shall see now. Thereto, we henceforth assume that $k$ is of the multiplicative form
\begin{equation}\label{kernel}
k(x,y)=\prod_{t=0}^T k_t(x_t,y_t)
\end{equation}
for measurable kernels $k_t$ on $E_t$ such that $\kappa_t\in L^2_{\Q_t}$ for $\kappa_t(x)=\sqrt{k_t(x,x)}$, and with separable RKHS $\Hcal_t$. The RKHS of $k$ can then be identified with the tensor product $\Hcal=\Hcal_0\otimes\cdots\otimes \Hcal_T$, see \cite[Theorem 5.11]{pau_rag_16}. In particular, $\langle g, h\rangle_\Hcal =\prod_{t=0}^T \langle g_t,h_t\rangle_{\Hcal_t}$ for functions $g(x)=\prod_{t=0}^T g_t(x_t)$ and $h(x)=\prod_{t=0}^T h_t(x_t)$.

It is then easy to see that the kernel $k$ in \eqref{kernel} is tractable if the kernel embeddings $m_t(y) =\int_{E_t} k_t(x,y) \Q_t(dx)$, see \cite{sri_etal_10}, are in closed form for all $y\in E_t$ and $t=0,\dots,T$. Indeed, the conditional kernel embeddings can now be written as
\begin{equation}\label{condKernEmb}
    M_t(y)=\E_\Q[k(X,y)\mid \Fcal_t] = \prod_{s=0}^t k_s(X_s,y_s)\prod_{s=t+1}^T m_{s}(y_s),\quad y\in E.
\end{equation}

We next assume that each $E_t$ is a measurable subset of $\R^{d_t}$ for some $d_t\in\N$. Then Bochner's theorem \cite[Proposition 2.5]{sat_99} implies that any symmetric probability measure $\Lambda$ on $\R^{d_t}$, and parameter $\beta\ge 0$, give rise to a kernel on $E_t$ of the form
\begin{equation}\label{TIK}
  k_t(x,y)=\e^{\beta x^\top y}\int_{\R^{d_t}} \e^{\im (x-y)^\top \lambda} \Lambda(d\lambda),\quad x,y\in E_t.\footnote{$\Lambda$ is symmetric if $\Lambda(-B) = \Lambda(B)$, where $-B= \{-x\mid x \in B\}$, for every Borel measurable set $B\subset \R^{d_t}$.}
\end{equation}
As for the random driver distribution, we assume that every $\Q_t$ is infinitely divisible and admits exponential moments of order $\beta x$, for all $x\in E_t$. Then the L\'evy--Khintchine formula yields a closed form expression for the (extended) characteristic function $\widehat \Q_t(u)=\int_{E_t} \e^{ u^\top y}\Q_t(dy)$ for all admissible $u\in\C^{d_t}$, see \cite[Theorem 8.1]{sat_99}. Examples include (discrete-time) L\'evy processes $X$, which are widespread stochastic drivers in financial models. The kernel embedding becomes
\begin{equation}\label{PTIK}
   m_{t}(x) = \int_{\R^{d_t}}\int_{E_t} \e^{(\beta x+\im \lambda)^\top y} \Q_t(dy) \e^{-\im x^\top \lambda}\Lambda(d\lambda) = \int_{\R^{d_t}}\widehat \Q_t(\beta x+\im \lambda) \e^{-\im x^\top \lambda}\Lambda(d\lambda),\quad x\in E_t,
\end{equation}
which is in closed form subject to an integration with respect to $\Lambda(d\lambda)$. In order to appreciate this finding, we note that Fourier type integrals like the one on the right hand side in \eqref{PTIK} are routinely computed, e.g, in L\'evy type or affine models, \cite{duf_fil_sch_03}. So we can draw on a large library of available computer code.

Tractable measures $\Lambda$ include symmetric infinitely divisible distributions, for which the L\'evy--Khintchine formula yields a closed form expression for $k_t$ in \eqref{TIK},
\[ k_t(x,y)=\e^{\beta x^\top y} \e^{-\frac{1}{2} (x-y)^\top A (x-y) + \int_{\R^{d_t}} (\cos((x-y)^\top\xi)-1)\nu(d\xi)},\quad x,y\in E_t,\]
where $A$ is a positive semi-definite matrix, and $\nu$ is a symmetric L\'evy measure on $\R^{d_t}$, see \cite[Theorem 8.1 and E 18.1]{sat_99}. Such kernels for $\beta=0$ have recently also been studied by \cite{nis_fuk_16}. For $\nu=0$ and $A = 2\alpha I_{d_t}$, where $I_{d_t}$ is the identity matrix, we obtain the \emph{Gaussian-exponentiated kernel}
\begin{equation}\label{eqGEK}
  k_t(x,y)= \e^{-\alpha\|x-y\|^2+\beta x^\top y},\quad x,y\in E_t,
\end{equation}
with parameters $\alpha\ge 0 $ and $\beta\ge 0$. This contains the Gaussian kernel, for $\beta=0$, and the exponentiated kernel, for $\alpha= 0$, as special cases.

Also the kernels of Sobolev spaces are of the form \eqref{TIK} with $\beta=0$. \cite{nov_etal_18} recently showed that the reproducing kernel of the Sobolev space $W^s_2(\R^{d_t})$ of functions whose weak derivatives up to order $s>d_t/2$ are square-integrable is given by the probability measure $\Lambda(d\lambda)=  (2\pi)^{-d_t} (1+\sum_{0<|\bm\alpha|\le s}\lambda^{\bm \alpha})^{-1} d\lambda$. This is noteworthy, as Sobolev spaces are versatile tools for function approximation, and thus potentially useful for tractable finance applications.

\section{Examples}\label{secexamples}

We extend on the introductory example with the Black--Scholes model with $d$ nominal stock price processes $S_{i,t}$ given by~\eqref{SeqBS}, for some dimension $d\in\N$. In particular, we assume that $X_t$ are i.i.d.\ standard Gaussians on $E_t=\R^d$, $t=1,\dots,T$.\footnote{Note that we do not specify $X_0$ here, which could include portfolio specific values that parametrize the cumulative cashflow function $f(X)$. This could include the strike price of an embedded option or the initial values of underlying financial instruments. We could sample $X_0$ from a Bayesian prior $\Q_0$. We henceforth omit $X_0$, which is tantamount to setting $k_0=1$.\label{footnote_X0}}

As for components of the kernel \eqref{kernel}, we consider the Gaussian-exponentiated kernels \eqref{eqGEK} with parameters $\alpha> 0$ and $\beta\in [0,1/2)$. The upper bound on $\beta$ is necessary and sufficient for \eqref{ass0} to hold. Whenever appropriate, we identify the path space $E$ with $\R^{dT}$ by stacking $x=(x_1,\dots,x_T)$ into a column vector. Accordingly, $\Q=\Ncal(0,I_{dT})$ is the standard Gaussian measure on $\R^{dT}$, and we can write $k(x, y) = \e^{-\alpha \|x-y\|^2+\beta x^\top y}$.

In view of Lemma~\ref{lemsepHnew}, every $h\in \Hcal$ is continuous and $\Hcal$ is separable. For the following important property we recall Definition~\ref{defL2UK}.
\begin{lemma}\label{lemGEKuniversalnew}
The Gaussian-exponentiated kernel $k$ is $L^2_\Q$-universal.
\end{lemma}

As for the sampling measure $\widetilde\Q$, we consider the Radon--Nikodym derivative $w=d{\widetilde\Q}/d\Q$ given by
\[  w(x)=  (1-2\gamma)^{dT/2}\e^{\gamma \|x\|^2} \]
with parameter $\gamma<1/2$. Then $\widetilde\Q=\Ncal(0, (1-2\gamma)^{-1} I_{d  T})$ is a centered Gaussian measure with scaled variance, so that \eqref{sampnu} is clearly satisfied. We obtain $\widetilde\kappa(x)= (1-2\gamma)^{-dT/4} \e^{(\beta/2-\gamma/2) \|x\|^2}$. Hence condition \eqref{newasstildek} holds if and only if
\begin{equation}
 \beta\le \gamma,\label{folp2}
\end{equation}
which we henceforth assume. Note that for $\beta=\gamma$ we obtain the Radon--Nikodym derivative \eqref{natwasnew}, which is optimal in the sense of Lemma~\ref{lemoptwnew}.

For the Gaussian sampling measure $\widetilde\Q$, \eqref{XiXjneq} almost surely holds for any finite sample, so that $\bar n=n$, $\bar X^{(j)}=X^{(j)}$ and $|I_j|=1$ for all $j=1,\dots,n$. This simplifies the expression of the estimator $V_{\bm X, t}$ in \eqref{hatVcf1}, which also involves the conditional kernel embeddings $M_t$, given in \eqref{condKernEmb}. Straightforward calculations show that the involved kernel embeddings are of the closed form
\begin{equation}\label{eqmsBS}
  m_s(y_s) = (1+2\alpha)^{-d/2}\e^{ \frac{\beta^2 + 4 \alpha \beta - 2 \alpha}{4 \alpha + 2} \|y_s\|^2}.
\end{equation}

As for the portfolios, we fix a strike price $K$ and consider the following European options with discounted payoff functions
\begin{itemize}
\item Min-put $f(X)=  \e^{-r\sum_{t=1}^T \Delta_t} (K-\min_{i} S_{i,T})^+$;
\item Max-call $f(X)=\e^{-r\sum_{t=1}^T \Delta_t} (\max_{i} S_{i,T} -  K)^+$.
\end{itemize}
We also consider a genuinely path-dependent product with the discounted payoff function
\begin{itemize}
\item Barrier reverse convertible $f(X)= \e^{-r\sum_{t=1}^T \Delta_t} \left(C  +  F  \left( 1 - 1_{\{\min_{i,t}   S_{i,t}\le B\}}   \left(1 - \min_{i} \frac{ S_{i,T}}{S_{i,0}  K} \right)^+\right)\right)$,
\end{itemize}
for some barrier $B<K$, coupon $C$, and face value $F$. At maturity $T$, the holder of this structured product receives the coupon $C$. She also receives the face value $F$ if none of the nominal stock prices falls below the barrier $B$ at any time step $t=1,\dots,T$. Otherwise, the face value $F$ is reduced by the payoff of $F/K$ min-puts on the normalized stocks $S_{i,T}/S_{i,0}$ with strike price $K$. These examples are inspired from those given in \cite{bec_etal_2019}. Note that the payoff functions of the min-put and barrier reverse convertible are bounded, while the payoff of the max-call is unbounded.

For our numerical experiments, we choose the following parameter values: risk-free rate $r=0$, initial stock prices $S_{i,0}=1$, volatilities $\sigma_i=0.2 \bm e_i$, where $\bm e_i$ denote the standard basis vectors in $\R^d$, so that stock prices are independent, strike price $K=1$ (at the money), barrier $B=0.6$, coupon $C=0$, and face value $F=1$. The remaining parameters are chosen case-by-case as follows:
\begin{itemize}
\item Min-put: $d=6$ stocks, $T=2$ time steps with step sizes $\Delta_1 = 1/12$ and $\Delta_2=11/12$, and sampling measure parameter $\gamma=0$. The last is justified as the min-put payoff is bounded. Note that necessarily $\beta=0$ by \eqref{folp2}. Hence condition \eqref{newasstildek}, and thus both Theorems \ref{thmLT} and \ref{thmcinew} hold.
\item Max-call: $d=6$, $T=2$, $\Delta_1=1/12$, $\Delta_2=11/12$, as for the min-put. However, condition \eqref{newasstildef} holds---and Theorem~\ref{thmcinew} applies---if and only if $\gamma>0$. On the other hand, in view of Remark~\ref{remthmLT}, Theorem~\ref{thmLT} still applies also for $\gamma=0$. So we try $\gamma=0$ and $\gamma =0.15$.\footnote{Numerical issues arise for $\gamma>0.15$. Indeed, the sample estimator of $\E_{\widetilde\Q}[1/w(X)]=1$ gives values that are significantly smaller than $1$, due to limited precision when representing sample values of $1/w(X)$ that are close to zero in dimension 36.}
\item Barrier reverse convertible: $d=3$ stocks, $T=12$ time steps with step sizes $\Delta_t=1/12$, and sampling measure parameter $\gamma=0$. The last is justified as for the min-put, and implies that both Theorems \ref{thmLT} and \ref{thmcinew} hold.
\end{itemize}

In \cite[Section~5]{fer_fil_2020}, an insurance liability model is constructed using models often used in practice, where $d=5$. For the min-put and max-call, we have $d=6$ and the dimension of the path space $E=\R^{dT}$ amounts to $12$; for the barrier reverse convertible these values are $3$ and $36$, respectively. In practical terms, these examples can thus be considered high-dimensional.

Under the parameter specifications above, we generate a training sample $\bm X$ of size $n= 2\times 10^4$ and use the Gaussian Process Regression (GPR) module of the scikit-learn library \cite{scikit_learn}. Indeed, GPR yields the same expression as we have for the sample estimator $f_{\bm X}$ in Lemma~\ref{lemcompnewX}, see \cite{ras_wil_2005}. The advantage of using GPR is that some optimal hyperparameter values $\alpha$, $\beta$ and $\lambda$ are obtained by maximizing a likelihood function \cite{ras_wil_2005}. This is an alternative to the standard validation step where one needs to specify a grid for every hyperparemeter, which can lead to cumbersome and lengthy computations, as we experienced for our examples. Instead, for GPR we only need to specify value ranges for each hyperparemeter, which here we chose as $\alpha \in [2.8\times 10^{-5}, 83]$, $\beta \in [10^{-9}, 0.15]$ and $\lambda \in [10^{-12}, 10^{-3}]$. Table \ref{tableHP} shows the optimal hyperparameter values. We notice that all optimal values lie inside their pre-specified ranges.

\begin{table}[h]
\centering
  \begin{tabular}{|l|c|c|c|}
  \hline
Payoff & $\alpha$  &  $\beta$ & $\lambda$   \\
\hline\hline
Min-put  & $2.06\times 10^{-2}$ & 0 & $1.86\times10^{-8}$\\
\hline
Max-call ($\gamma=0$)  & $2.53\times 10^{-2}$ & 0 & $3.33\times 10^{-8}$\\
 Max-call ($\gamma=0.15$)  & $3.66\times 10^{-2}$ & $3.25\times 10^{-9}$ & $4.14\times 10^{-8}$\\
\hline
Barrier reverse convertible & $2.96\times10^{-3}$ & 0 & $9.20\times 10^{-8}$\\
\hline
\end{tabular}
\caption{Optimal hyperparameter values $\alpha$, $\beta$, $\lambda$ from GPR.}\label{tableHP}
\end{table}

We then compute the estimated value process $V_{\bm X,t}$ at time steps $t \in \{0, 1, T\}$ using Lemma~\ref{lemcompnewX}, \eqref{condKernEmb} and \eqref{eqmsBS}. We benchmark $V_{\bm X}$ to the ground truth value process $V$, which we obtain by means of large Monte Carlo schemes using $n_{test} = 10^5$ simulations. More specifically, we obtain $V_0$ as simple Monte Carlo estimate from simulating $V_T=f(X)$. For $V_1$, we use a nested Monte Carlo scheme, where we estimate each sample of $V_1=V_1(X_1)$ using $n_{inner} = 1000$ independent inner simulations of $(X_2,\dots,X_T)$. Then we carry out the following computations.

First, we compute the absolute relative error of $V_{\bm X, 0}$, $|V_{\bm X, 0} - V_0|/V_0$, and the normalized $L^2_{\Q}$-errors of $V_{\bm X, t}$, $\|V_{\bm X, t} - V_t\|_{2, \Q}/V_0$, for $t=1, T$. Table \ref{table1new} shows that the normalized $L^2_\Q$-error of $V_{\bm X, t}$ decreases substantially with the time-to-maturity $T-t$. More specifically, the normalized $L^2_\Q$-error of $V_{\bm X, 1}$ is, on average, 10-times smaller than that of $V_{\bm X, T}$. And the relative absolute error of $V_{\bm X, 0}$ is, on average, 19-times smaller than the normalized $L^2_\Q$-error of $V_{\bm X, 1}$. These findings are in line with \eqref{doobineq} and have useful practical implications. Indeed, the sample error bound in Theorem~\ref{thmcinew} is, arguably, mainly of theoretical interest and hardly available in practice. However, in concrete applications, one can always estimate the normalized $L^2_{\Q}$-error of $V_{\bm X, T}$ by a simple Monte Carlo scheme as we do here. This error then serves as upper bound on the normalized $L^2_{\Q}$-errors of $V_{\bm X, t}$, for any $t<T$. Figures \ref{error_V1_minput}, \ref{error_V1_maxcall}, \ref{error_V1_barrier} and Figures \ref{error_VT_minput}, \ref{error_VT_maxcall}, \ref{error_VT_barrier} show the decrease of the normalized $L^2_\Q$-errors with respect to the training sample size $n$ for $V_{\bm X, 1}$ and $V_{\bm X, T}$, respectively.

\begin{table}[h]
\centering
  \begin{tabular}{|l|r|r|r|r|r|}
\hline
Payoff &    $V_{\bm X, 0}$ &        $V_{\bm X, 1}$& $V_{\bm X, T}$\\
\hline\hline
Mint-put & 0.1942&  1.827& 10.05   \\
\hline
Max-call $(\gamma = 0)$ & 0.07962& 2.500& 12.35\\
Max-call $(\gamma = 0.15)$ &  0.1031& 2.315 & 11.65\\
\hline
Barrier reverse convertible &0.02198& 0.2506& 5.745\\
\hline
\end{tabular}
\caption{Normalized $L^2_\Q$-error $\|V_t- V_{\bm X, t}\|_{2,\Q}/V_0$ at steps $t=0, 1, T$. All values are expressed in \%.}\label{table1new}
\end{table}

\begin{figure}[p]
    \centering 
\begin{subfigure}{0.45\textwidth}
  \includegraphics[width=\linewidth]{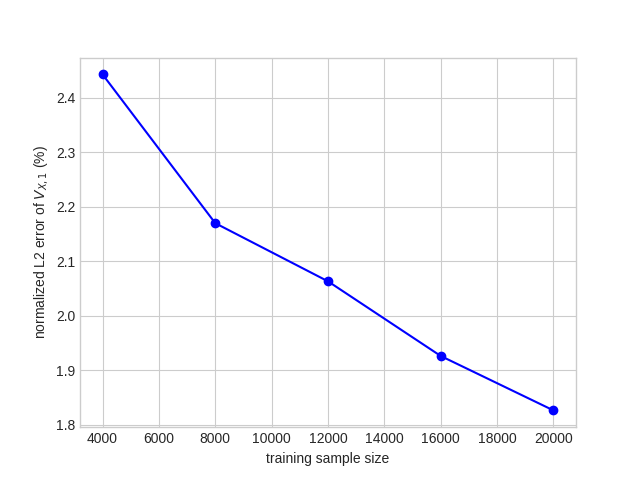}
  \caption{Normalized $L^2_{\Q}$-error of $V_{\bm X, 1}$ in \%}
  \label{error_V1_minput}
\end{subfigure}\hfil 
\begin{subfigure}{0.45\textwidth}
  \includegraphics[width=\linewidth]{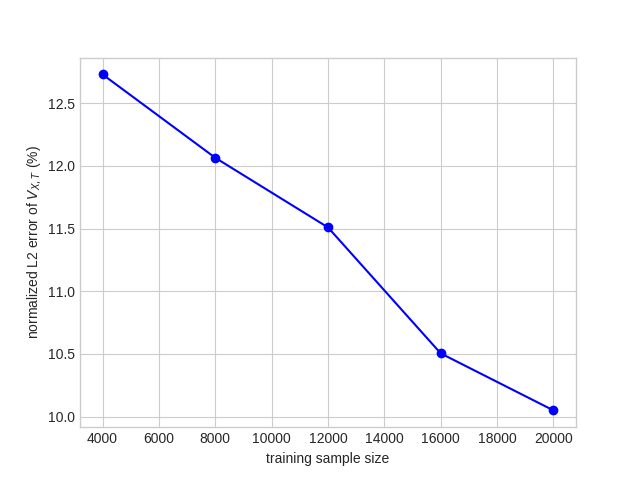}
  \caption{Normalized $L^2_{\Q}$-error of $V_{\bm X, T}$ in \%}
  \label{error_VT_minput}
\end{subfigure}
\medskip
\begin{subfigure}{0.45\textwidth}
  \includegraphics[width=\linewidth]{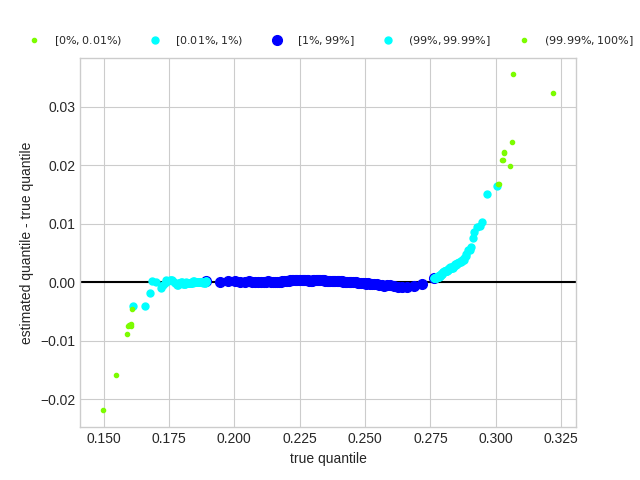}
  \caption{Detrended Q-Q plot of $V_{\bm X, 1}$}
  \label{qqplot_V1_minput}
\end{subfigure}\hfil 
\begin{subfigure}{0.45\textwidth}
  \includegraphics[width=\linewidth]{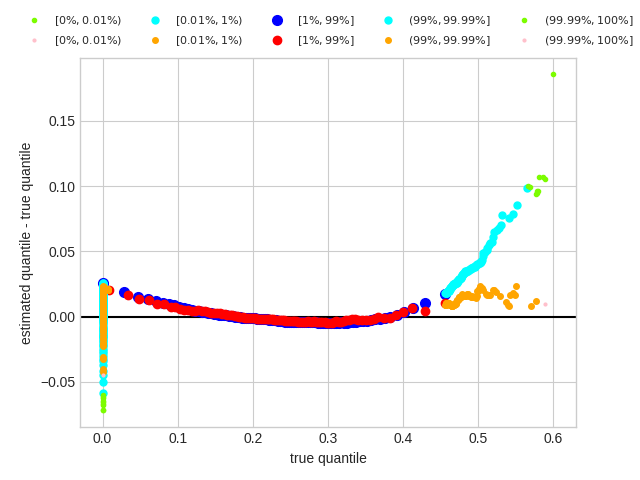}
  \caption{Detrended Q-Q plot of $V_{\bm X, T}$}
  \label{qqplot_VT_minput}
\end{subfigure}
\caption{Results for the min-put. In the detrended Q-Q plots, the blue, cyan, and lawngreen (red, orange, and pink) dots are built using the test (training) data. $[0\%, 0.01\%)$ refer to the quantiles of levels $\{0.001\%,0.002\%, \cdots, 0.009\%\}$, $[0.01\%, 1\%)$ refer to the quantiles of levels $\{0.01\%,0.02\%, \cdots, 0.99\%\}$, $[1\%, 99\%]$ refer to the quantiles of levels $\{1\%,2\%, \cdots, 99\%\}$, $(99\%, 99.99\%]$ refer to the quantiles of levels $\{99.01\%,99.02\%, \cdots, 99.99\%\}$, and $(99.99\%, 100\%]$ refer to the quantiles of levels $\{99.991\%,99.992\%, \cdots, 100\%\}$. }
\label{fig_minput}
\end{figure}

\begin{figure}[p]
    \centering 
\begin{subfigure}{0.45\textwidth}
  \includegraphics[width=\linewidth]{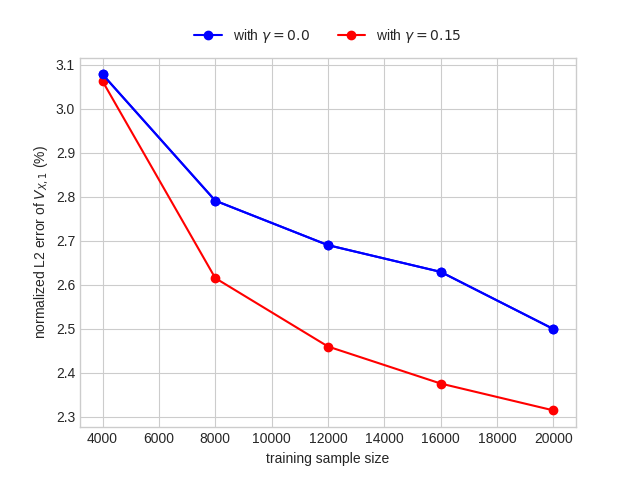}
  \caption{Normalized $L^2_{\Q}$-error of $V_{\bm X, 1}$ in \%}
  \label{error_V1_maxcall}
\end{subfigure}\hfil 
\begin{subfigure}{0.45\textwidth}
  \includegraphics[width=\linewidth]{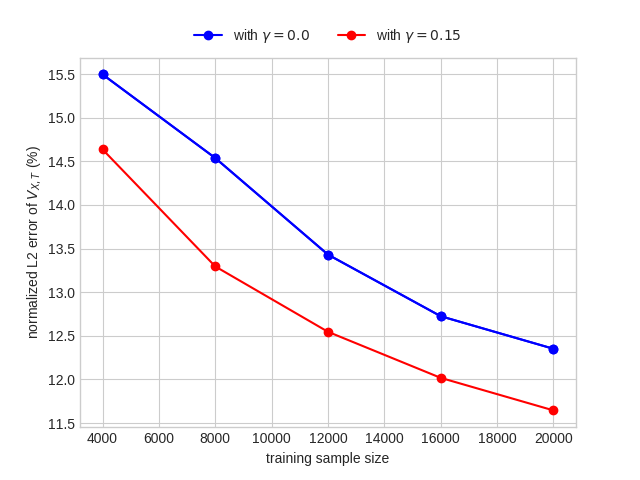}
  \caption{Normalized $L^2_{\Q}$-error of $V_{\bm X, T}$ in \%}
  \label{error_VT_maxcall}
\end{subfigure}
\medskip
\begin{subfigure}{0.45\textwidth}
  \includegraphics[width=\linewidth]{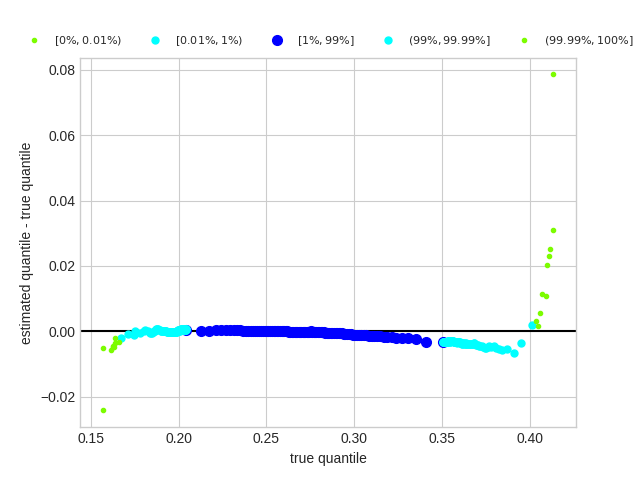}
  \caption{Detrended Q-Q plot of $V_{\bm X, 1}$ for $\gamma = 0$}
  \label{qqplot_V1_maxcall000}
\end{subfigure}\hfil 
\begin{subfigure}{0.45\textwidth}
  \includegraphics[width=\linewidth]{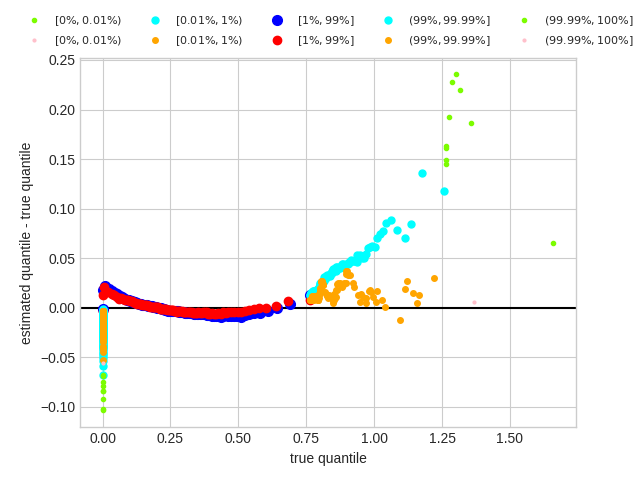}
  \caption{Detrended Q-Q plot of $V_{\bm X, T}$ for $\gamma = 0$}
  \label{qqplot_VT_maxcall000}
\end{subfigure}
\medskip
\begin{subfigure}{0.45\textwidth}
  \includegraphics[width=\linewidth]{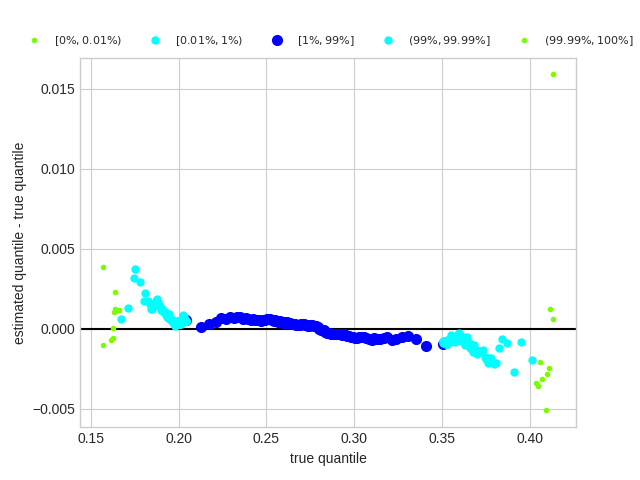}
  \caption{Detrended Q-Q plot of $V_{\bm X, 1}$ for $\gamma = 0.15$}
  \label{qqplot_V1_maxcall015}
\end{subfigure}\hfil 
\begin{subfigure}{0.45\textwidth}
  \includegraphics[width=\linewidth]{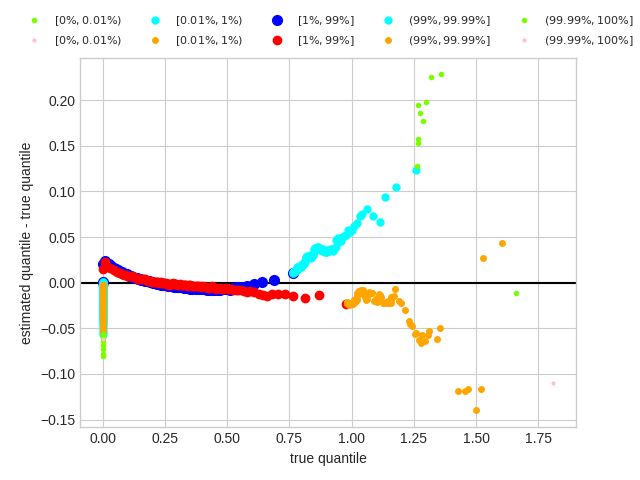}
  \caption{Detrended Q-Q plot of $V_{\bm X, T}$ for $\gamma = 0.15$}
  \label{qqplot_VT_maxcall015}
\end{subfigure}
\caption{Results for the max-call with $\gamma = 0.0$ and $\gamma = 0.15$. In the detrended Q-Q plots, the blue, cyan, and lawngreen (red, orange, and pink) dots are built using the test (training) data. $[0\%, 0.01\%)$ refer to the quantiles of levels $\{0.001\%,0.002\%, \cdots, 0.009\%\}$, $[0.01\%, 1\%)$ refer to the quantiles of levels $\{0.01\%,0.02\%, \cdots, 0.99\%\}$, $[1\%, 99\%]$ refer to the quantiles of levels $\{1\%,2\%, \cdots, 99\%\}$, $(99\%, 99.99\%]$ refer to the quantiles of levels $\{99.01\%,99.02\%, \cdots, 99.99\%\}$, and $(99.99\%, 100\%]$ refer to the quantiles of levels $\{99.991\%,99.992\%, \cdots, 100\%\}$. }
\label{fig_maxcall_highd}
\end{figure}

\begin{figure}[p]
    \centering 
\begin{subfigure}{0.45\textwidth}
  \includegraphics[width=\linewidth]{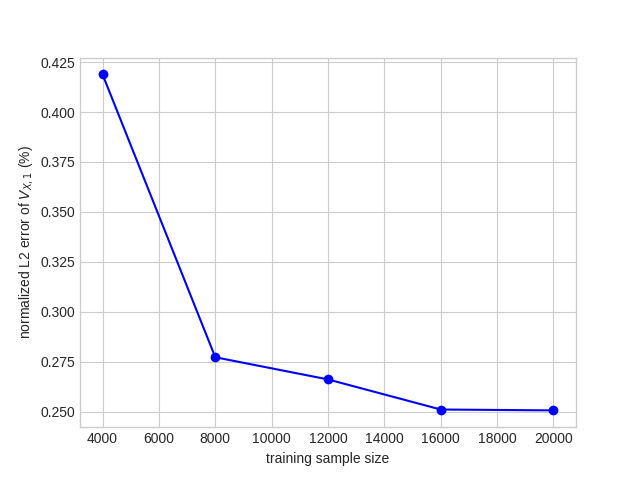}
  \caption{Normalized $L^2_{\Q}$-error of $V_{\bm X, 1}$ in \%}
  \label{error_V1_barrier}
\end{subfigure}\hfil 
\begin{subfigure}{0.45\textwidth}
  \includegraphics[width=\linewidth]{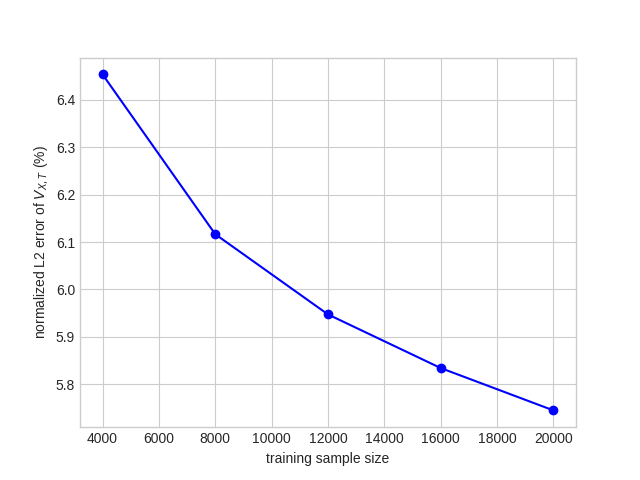}
  \caption{Normalized $L^2_{\Q}$-error of $V_{\bm X, T}$ in \%}
  \label{error_VT_barrier}
\end{subfigure}
\medskip
\begin{subfigure}{0.45\textwidth}
  \includegraphics[width=\linewidth]{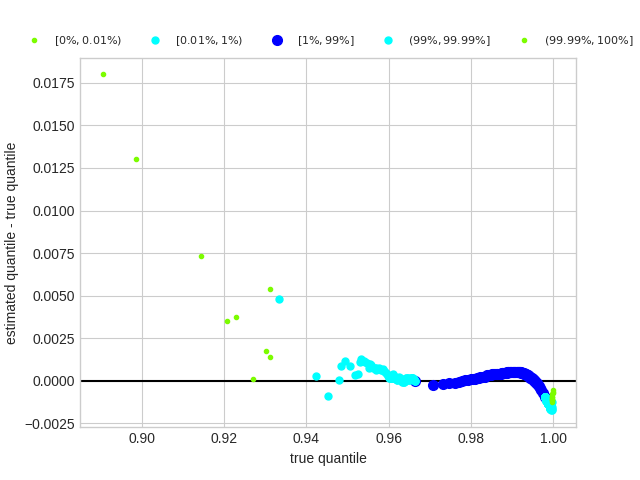}
  \caption{Detrended Q-Q plot of $V_{\bm X, 1}$}
  \label{qqplot_V1_barrier}
\end{subfigure}\hfil 
\begin{subfigure}{0.45\textwidth}
  \includegraphics[width=\linewidth]{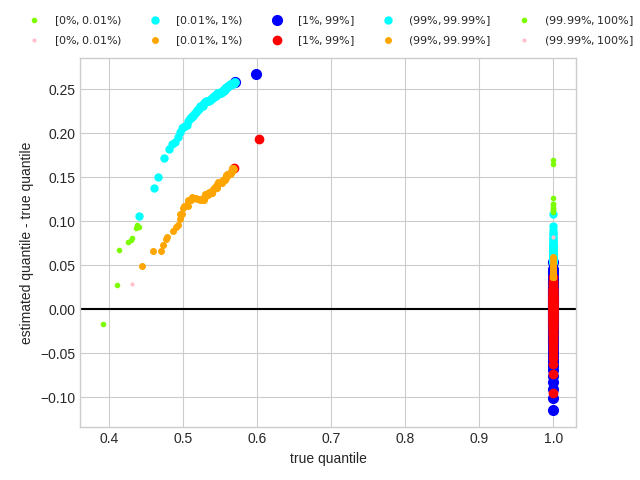}
  \caption{Detrended Q-Q plot of $V_{\bm X, T}$}
  \label{qqplot_VT_barrier}
\end{subfigure}
\caption{Results for the barrier reverse convertible. In the detrended Q-Q plots, the blue, cyan, and lawngreen (red, orange, and pink) dots are built using the test (training) data. $[0\%, 0.01\%)$ refer to the quantiles of levels $\{0.001\%,0.002\%, \cdots, 0.009\%\}$, $[0.01\%, 1\%)$ refer to the quantiles of levels $\{0.01\%,0.02\%, \cdots, 0.99\%\}$, $[1\%, 99\%]$ refer to the quantiles of levels $\{1\%,2\%, \cdots, 99\%\}$, $(99\%, 99.99\%]$ refer to the quantiles of levels $\{99.01\%,99.02\%, \cdots, 99.99\%\}$, and $(99.99\%, 100\%]$ refer to the quantiles of levels $\{99.991\%,99.992\%, \cdots, 100\%\}$. }
\label{fig_barrier}
\end{figure}

Second, we draw the detrended Q-Q plots of $V_{\bm X, 1}$ and $V_{\bm X, T}$, using the $n_{test}$ test samples and $n$ training samples, respectively. Thereto, we compute the empirical left quantiles of $V_{\bm X, t}$ and $V_t$ at the levels of $\{0.001\%, 0.002\%, \dots, 0.009\%\}$, $\{0.01\%, 0.02\%,\dots, 0.99\%\}$, $\{1\%,2\%,\dots,99\%\}$, $\{99.01\%, 99.02\%, \dots, 99.99\%\}$, and $\{99.991\%, 99.992\%,\cdots,  100\%\}$.\footnote{Note that for the test sample of size $n_{test}=10^5$, the left $0.001\%$-quantile ($100\%$-quantile) corresponds to the smallest (largest) sample value. For the training sample of size $n_{train}=2\times 10^4$, the same holds, while the ten left- and right-most quantiles collapse to two values, respectively.} The detrended quantiles (difference between estimated quantiles minus true quantiles) are then plotted against the true quantiles. Figures \ref{qqplot_V1_minput}, \ref{qqplot_V1_maxcall000}, \ref{qqplot_V1_maxcall015}, \ref{qqplot_V1_barrier} and Figures \ref{qqplot_VT_minput}, \ref{qqplot_VT_maxcall000}, \ref{qqplot_VT_maxcall015}, \ref{qqplot_VT_barrier} show the detrended Q-Q plots of $V_{\bm X, 1}$ and $V_{\bm X, T}$, respectively. We observe that the detrended Q-Q plot of $V_{\bm X, 1}$ is significantly better than that of $V_{\bm X, T}$, which is in line with our previous findings for the corresponding relative $L^2_\Q$-errors. Notably, Figure \ref{qqplot_VT_barrier} reveals that for less than $3\%$ (as indicated by the two leftmost red dots) of the training sample (that is, less than $600$ points out of $n= 20{,}000$) the embedded min-put options in the barrier reverse convertible are triggered and in the money. For the remaining sample points the payoff is equal to the face value, $F=1$. And yet, as Figure \ref{qqplot_V1_barrier} shows, this is enough for our algorithm to learn the payoff function such that $V_{\bm X, 1}$ is remarkably close to the ground truth with a normalized $L^2_\Q$-error of $0.251\%$, as reported in Table \ref{table1new}. Figure \ref{fig_maxcall_highd} shows the benefit in using $\gamma>0$ over $\gamma = 0$ for the unbounded payoff of the max-call, which is consistent with Theorem \ref{thmcinew}. We also computed the normalized $L^2_\Q$-errors and detrended Q-Q plots for min-put and barrier reverse convertible with $\gamma = 0.15$, and we found slightly better, unreported, results than with $\gamma = 0$, which are available from the authors upon request. We expect that our results can be further improved by choosing the sampling measure $\widetilde\Q\sim\Q$ more tailored to the specific underlying portfolio payoff, leading to more balanced training samples. We leave this up for future research.

Third, as risk management application, we compute the value at risk and expected shortfall of long and short positions in the above portfolios. Thereto, we recall the definitions that can also be found in \cite[Chapter~4]{foe_sch_04}. For a confidence level $\alpha \in (0,1)$, the value at risk is defined as left $\alpha$-quantile of the loss distribution, $\mathrm{VaR}_\alpha(L)=\inf\{y \mid \Pa[L\le y] \ge \alpha\}$, and the expected shortfall is given by $\mathrm{ES}_{\alpha}(L) = \frac{1}{1-\alpha}\E_\Pa[(L-q_\alpha)^+] + q_\alpha$, where $q_\alpha$ is an $\alpha$-quantile of $L$, e.g., $q_\alpha=\mathrm{VaR}_\alpha(L)$. Both value at risk and expected shortfall are standard risk measures in practice. For instance, insurance companies have to compute the value at risk at level $\alpha=99.5\%$ and the expected shortfall at level $\alpha=99\%$, under Solvency II and the Swiss Solvency Test, respectively. For more discussion on these two risk measures we refer to the book \cite{mcn_etal_15}. Henceforth, we assume the real-world measure $\Pa=\Q$, for simplicity. For the three above examples, we compute normalized value at risk and expected shortfall of the 1-period loss $\mathrm{L}=V_{0}-V_{1}$ and its estimator $\mathrm{L}_{\bm X}=V_{\bm X, 0}-V_{\bm X, 1}$ of a long position, namely $\mathrm{VaR}_{99.5\%}(\mathrm{L})/V_0$, $\mathrm{ES}_{99\%}(\mathrm{L})/V_0$, $\mathrm{VaR}_{99.5\%}(\mathrm{L}_{\bm X})/V_0$, and $\mathrm{ES}_{99\%}(\mathrm{L}_{\bm X})/V_0$. We compute the same risk measures for a short position, namely $\mathrm{VaR}_{99.5\%}(-\mathrm{L})/V_0$, $\mathrm{ES}_{99\%}(-\mathrm{L})/V_0$, $\mathrm{VaR}_{99.5\%}(-\mathrm{L}_{\bm X})/V_0$, and $\mathrm{ES}_{99\%}(-\mathrm{L}_{\bm X})/V_0$. Tables \ref{table2new} and \ref{table3new} show that risk measure estimates of the long positions are strikingly accurate. Risk measure estimates of the short positions are less good. However, note that these risk measures are a tough metric for our estimators because they focus on the tails of the distribution beyond the $1\%$- and $99\%$-quantiles, respectively. All these observations are in line with the detrended Q-Q plots discussed above. In fact, in Figures \ref{qqplot_V1_minput}, \ref{qqplot_V1_maxcall000}, \ref{qqplot_V1_barrier} we see that our method gives a better estimation of the left tail distribution than the right tail distribution. Also note the benefit of choosing the sampling measure $\widetilde \Q$ over $\Q$. In fact, all risk measurements, except for the expected shortfall of $\mathrm{L}_{\bm X}$, are more accurate when the sampling is from $\widetilde\Q$ than when it is from $\Q$.

\begin{table}[h]
\centering
  \begin{tabular}{|l|r|r|r|r|}
\hline
Payoff &    $\mathrm{VaR}_{99.5\%}(\mathrm{L})$ & $\mathrm{VaR}_{99.5\%}(\mathrm{L}_{\bm X})$& $\mathrm{VaR}_{99.5\%}(-\mathrm{L})$ & $\mathrm{VaR}_{99.5\%}(-\mathrm{L}_{\bm X})$\\
\hline\hline
Mint-put & 2063& 2083 & 2058 & 2123   \\
\hline
Max-call $(\gamma = 0)$ & 2800& 2802& 3071 & 2961\\
Max-call $(\gamma = 0.15)$ &  2800& \bf{2801} & 3071 &\bf{3041}\\
\hline
Barrier revere convertible &264.1& 264.6& 99.83 & 85.94\\
\hline
\end{tabular}
\caption{Normalized true and estimated value at risk $\mathrm{VaR}_{99.5\%}(\mathrm{L})/V_0$, $\mathrm{VaR}_{99.5\%}(\mathrm{L}_{\bm X})/V_0$, $\mathrm{VaR}_{99.5\%}(-\mathrm{L})/V_0$, and $\mathrm{VaR}_{99.5\%}(-\mathrm{L}_{\bm X})/V_0$. All values are expressed in basis points.}\label{table2new}
\end{table}

\begin{table}[h]
\centering
  \begin{tabular}{|l|r|r|r|r|}
\hline
Payoff &    $\mathrm{ES}_{99\%}(\mathrm{L})$ & $\mathrm{ES}_{99\%}(\mathrm{L}_{\bm X})$& $\mathrm{ES}_{99\%}(-\mathrm{L})$ & $\mathrm{ES}_{99\%}(-\mathrm{L}_{\bm X})$\\
\hline\hline
Mint-put & 2141& 2168& 2118 & 2219   \\
\hline
Max-call $(\gamma = 0)$ & 2890& \bf{2880}& 3205 & 3090\\
Max-call $(\gamma = 0.15)$ &  2890& 2870 & 3205 &\bf{3160}\\
\hline
Barrier revere convertible &284.7&283.2&101.2&86.63\\
\hline
\end{tabular}
\caption{Normalized true and estimated expected shortfall $\mathrm{ES}_{99\%}(\mathrm{L})/V_0$, $\mathrm{ES}_{99\%}(\mathrm{L}_{\bm X})/V_0$, $\mathrm{ES}_{99\%}(-\mathrm{L})/V_0$, and $\mathrm{ES}_{99\%}(-\mathrm{L}_{\bm X})/V_0$. All values are expressed in basis points.}\label{table3new}
\end{table}

As another risk management application, we now sketch how to compute the $\Q$-variance optimal hedging strategies \eqref{psieqhedge} for the above examples, as outlined in Section \ref{secintro}. The tradable hedging instruments here would naturally be the underlying stocks with discounted value processes $G_{i,t}= \e^{-r\sum_{s=1}^t\Delta_s} S_{i,t}$. The gains, in view of \eqref{SeqBS}, accordingly are given by $\Delta G_{i,t}=G_{i,t-1}\left( \exp[ \sigma_i^\top X_t \sqrt{{\Delta_t}} - \|\sigma_i\|^2 {\Delta_t}/2]-1\right)$. It straightforward to compute the ingredients that give the approximate hedging strategy \eqref{psieqhedge}, where we replace $\Delta V_t$ by $\Delta V_{\bm X,t}$. First, we have $\E_\Q[ \Delta G_{i,t}  \Delta G_{j,t} \mid\Fcal_{t-1}]=G_{i,t-1}G_{j,t-1}\left( \exp[ \sigma_i^\top \sigma_j \Delta_t]-1\right)$. Second, $\E_\Q[ \Delta G_{i,t} \Delta V_{\bm X,t}\mid\Fcal_{t-1}]$ can readily be computed in closed form using Lemma~\ref{lemcompnewX}, \eqref{condKernEmb} and \eqref{eqmsBS}. For the sake of brevity, we leave the full hedging implementations for future research.

As for the scalability of our method, we conducted similar experiments as those presented above with larger values of $d$ and $T$. Results suggested that a training sample of size $n=20{,}000$ is not enough to deal with problems of dimension $d\times T\ge 60$ since the normalized $L^2_\Q$-errors of $V_{\bm X, 1}$ were higher than $5\%$.
Solutions to deal with high dimensional problems with kernel methods exist in the literature and they include random projections such as Nystr\"om approximation \cite{smo_sch_00, wil_see_01} and random features \cite{rah_rec_07}. We tried a recent algorithm \cite{rud_etal_17}, which is based on Nystr\"om approximation, however the normalized $L^2_\Q$-errors of $V_{\bm X, 1}$ were still higher than $5\%$. These results suggest that without a larger training sample and more computational resources, it may be difficult to reach a smaller normalized $L^2_\Q$-error for dimensions $d\times T\ge 60$.

\section{Conclusion}\label{secconc}

We introduce a unified framework for quantitative portfolio risk management, based on the dynamic value process of the portfolio. We approximate and learn the value process from a finite sample of the cumulative cash flow of the portfolio using kernel methods. Thereto we deploy the theory of reproducing kernel Hilbert spaces, which we find suitable for the learning of functions using simulated samples. We exploit tractable kernels in conjunction with the kernel representer theorem to obtain the sample estimator of the value process in closed form. We show asymptotic consistency and derive finite sample error bounds, which have been established in the previous literature only under regularity and boundedness assumptions on the target function that do not hold for finance applications in general. Numerical experiments for exotic, path-dependent options in the multivariate Black--Scholes model in large dimensions show good results for a moderate training sample size.

Our approach can be extended in various directions. One direction is to further develop the above examples to be deployed for production. This includes the full implementation of the sketched hedging strategies in particular. Another direction is to further explore the scalability of the presented methods to higher dimensional sample spaces. There is a large activity in the machine learning research that addresses the scalability of kernel methods. New findings could also benefit applications in portfolio valuation and risk management. This is ongoing research. A third direction is to value Bermudan options, see, e.g., \cite{lam_lap_2011}. In quantitative finance this is a challenging problem and numerical methods are required to estimate the optimal value process, see, e.g., \cite[Introduction]{bec_etal_2019} for references to several of these methods. The approach we developed in this paper can also be applied to deal with such a problem. In this case, our approach falls in the class of ``regress-later'' methods presented in \cite{gla_yu_04}. In fact, in ``regress-later'' the value functions are estimated by a projection onto a finite number of basis functions, whereas with our method they would be estimated by kernel ridge regressions. Our approach would yield closed form estimators of the value process and finite sample guarantees.

\begin{appendix}

\clearpage

\section{Some facts about Hilbert spaces}\label{secfactsH}

For the convenience of the reader we collect here some basic definitions and facts about Hilbert spaces, on which our framework builds. We first recall some basics. We then introduce kernels and reproducing kernel Hilbert spaces. We then review compact operators and random variables on separable Hilbert spaces. For more background, we refer to, e.g., the textbooks \cite{kat_95,cuc_zho_07,ste_chr_08,pau_rag_16}.

\subsection{Basics}

We start with briefly recalling some elementary facts and conventions for Hilbert spaces. Let $H$ be a Hilbert space and $\Ical$ some (not necessarily countable) index set. We call a set $\{\phi_i\mid i\in \Ical\}$ in $H$ an \emph{orthonormal system (ONS)} in $H$ if $\langle  \phi_i,\phi_j\rangle_H=\delta_{ij}$, for the Kronecker Delta $\delta_{ij}$. We call $\{\phi_i\mid i\in \Ical\}$ an \emph{orthonormal basis (ONB)} of $H$ if it is an ONS whose linear span is dense in $H$. In this case, for every $h\in H$, we have $h=\sum_{i\in\Ical} \langle h,\phi_i\rangle_H \phi_i$ and the Parseval identify holds, $\|h\|_H^2 =\sum_{i\in\Ical} |\langle h,\phi_i\rangle_H|^2$, where only a countable number of coefficients $\langle h,\phi_i\rangle_H$ are different from zero. Here we recall the elementary fact that the closure of a set $A$ in $H$ is equal to the set of all limit points of sequences in $A$, see \cite[Theorem 2.37]{ali_bor_99}.

\subsection{Reproducing kernel Hilbert spaces}\label{ssecRKHS}

Let $k:E\times E\to \R$ be a kernel with RKHS $\Hcal$, as introduced at the beginning of Section~\ref{secapprox}. We collect some basic facts that are used in the paper.

The following lemma gives some useful representations of $k$, see \cite[Theorems 2.4 and 12.11]{pau_rag_16}, which hold for arbitrary set $E$.
\begin{lemma}\label{thmmercernew}
\begin{enumerate}
  \item\label{thmmercernew1} Let $\{\phi_i\mid i\in \Ical\}$ be an ONB of $\Hcal$. Then $k(x,y)=\sum_{i\in\Ical}\phi_i(x)\phi_i(y)$ where the series converges pointwise.
  \item\label{thmmercernew2} There exists a stochastic process $\phi_\omega(x)$, indexed by $x\in E$, on some probability space $(\Omega,\Fcal,\M)$ such that $\omega \mapsto\phi_\omega(x):\Omega\to\R$ are square-integrable random variables and $k(x,y) = \int_\Omega \phi_\omega(x) \phi_\omega(y)\,d\M(\omega)$.
\end{enumerate}
\end{lemma}

The next lemma provides sufficient conditions for continuity of the functions in $\Hcal$ and separability of $\Hcal$.
\begin{lemma}\label{lemsepHnew}
  Assume $(E,\tau)$ is a topological space. Then the following hold:
\begin{enumerate}
  \item\label{lemsepHnew1} If $k$ is continuous at the diagonal in the sense that
  \begin{equation}\label{contdiag}
       \text{$\lim_{y\to x} k(x,y) =\lim_{y\to x} k(y,y) = k(x,x)$ for all $x\in E$,}
      \end{equation}
then every $h\in\Hcal$ is continuous.

  \item\label{lemsepHnew2} If every $h\in\Hcal$ is continuous and $(E,\tau)$ is separable, then $\Hcal$ is separable.
\end{enumerate}
\end{lemma}

\begin{proof}
{\ref{lemsepHnew1}}: Let $h\in \Hcal$. Then $|h(x)-h(y)| \le \| k(\cdot,x)-k(\cdot,y)\|_{\Hcal}\|h\|_{\Hcal}$, by \eqref{eqfundamentalpnew}, with $ \| k(\cdot,x)-k(\cdot,y)\|_{\Hcal}= (k(x,x)-2 k(x,y)+k(y,y))^{1/2}$, and \eqref{contdiag} implies that $h$ is continuous.

{\ref{lemsepHnew2}}: This follows from \cite[Theorem~15]{ber_tho_2004}.
\end{proof}

\subsection{Compact operators on Hilbert spaces}\label{seccopoH}

Let $H,H'$ be separable Hilbert spaces. A linear operator (or simply an operator) $T:H\to H'$ is \emph{compact} if the image $(Th_n)_{n\ge 1}$ of any bounded sequence $(h_n)_{n \ge 1}$ of $H$ contains a convergent subsequence.

An operator $T:H\to H'$ is \emph{Hilbert--Schmidt} if $\|T\|_2=(\sum_{i\in I} \|T \phi_i\|^2_{H'})^{1/2} <\infty$, and \emph{trace-class} if $\|T\|_1=\sum_{i\in I} \langle (T^\ast T)^{1/2}\phi_i, \phi_i \rangle_H  <\infty$, for some (and thus any) ONB $\{\phi_i\mid i\in I\}$ of $H$. We denote by $\|T\|=\sup_{h\in H\setminus\{0\}} \|Th\|_{H'}/\|h\|_H$ the usual operator norm. We have $\|T\|\le\|T\|_2\le \|T\|_1$, thus trace-class implies Hilbert--Schmidt, and every Hilbert--Schmidt operator is compact.

A self-adjoint operator $T:H\to H$ is \emph{nonnegative} if $\langle Th, h \rangle_H \ge 0$, for all $h \in H$. Let $T:H\to H$ be a nonnegative, self-adjoint, compact operator. Then there exists an ONS $\{\phi_i\mid i\in I\}$, for a countable index set $I$, and eigenvalues ${\mu}_i>0$ such that the \emph{spectral representation} holds: $T = \sum_{i\in I} {\mu}_i\langle\cdot,\phi_i\rangle_\Hcal \phi_i$ .

\subsection{Random variables in Hilbert spaces}\label{appLTH}

Let $H$ be a separable Hilbert space and $\Q$ be a probability measure on $H$. The characteristic function $\widehat \Q:H\to\C$ of $\Q$ is defined by $\widehat \Q(h) = \int_H \e^{i \langle y,h \rangle_H} \Q(dy)$, $h \in H$.

If $\int_H \|y\|_H \Q(dy) < \infty$, then the mean $m_\Q=\int_H y  \Q(dy)$ of $\Q$ is well defined,  where the integral is in the Bochner sense, see, e.g., \cite[Section 1.1]{da_prato}. If $\int_H \|y\|^2_H \Q(dy) < \infty$, then the covariance operator $Q_\Q$ of $\Q$ is defined by $\langle Q_\Q h_1, h_2  \rangle_H = \int_H \langle y, h_1  \rangle_H \langle y, h_2  \rangle_H \Q(dy) - \langle m_\Q,h_1\rangle_H \langle m_\Q,h_2\rangle_H$, $h_1, h_2 \in H$. Hence $Q_\Q$ is a nonnegative, self-adjoint, trace-class operator. The measure $\Q$ is \emph{Gaussian}, $\Q\sim\Ncal(m_\Q,Q_\Q)$, if $\widehat{\Q}(h) = \e^{i \langle m_\Q,h \rangle_H - \frac{1}{2} \langle Q_\Q h, h \rangle_H}$, see \cite[Section 2.3]{da_prato}.

Now let $(\Omega, \Fcal, \mathbb{P})$ a probability space, and $(Y_n)_{n \ge 1}$ a sequence of i.i.d.\ $H$-valued random variables with distribution $Y_1 \sim \Q$. Assume that $\E[Y_1]=0$. If $\mathbb{E}[\|Y_1\|^2_H] < \infty$, then  $(Y_n)_{n \ge 1}$ satisfies the following \emph{law of large numbers}, see \cite[Theorem~2.1]{hoffmann_pisier},
\begin{equation}\label{Hlln}
    \frac{1}{n} \sum_{i= 1}^n Y_i \xrightarrow{a.s.} 0,
\end{equation}
and the \emph{central limit theorem}, see \cite[Theorem~3.6]{hoffmann_pisier},
\begin{equation}\label{Hclt}
\frac{1}{\sqrt{n}} \sum_{i = 1 }^n Y_i \xrightarrow{d}  \mathcal{N}(0,Q_\Q).
\end{equation}
If $\|Y_1\|_H \le 1$ a.s., then $(Y_n)_{n \ge 1}$ satisfies the following concentration inequality, called the \emph{Hoeffding inequality}, see \cite[Theorem 3.5]{pin_94},
\begin{equation}\label{hoeffdingeq}
    \Pa\left[\left\|\frac{1}{n} \sum_{i = 1}^n Y_i \right\|_\Hcal \ge \tau\right] \le 2 \e^{-\frac{\tau^2 n}{2}},\quad \tau>0.
\end{equation}

\section{Proofs}\label{secproofs}

We collect here all proofs from the main text.

\subsection{Properties of the embedding operator}\label{ssecHS}

For completeness, we first recall some basic properties of the operator $J$ defined in Section~\ref{secapprox}, which are used throughout the paper.

The operator $JJ^\ast$ is clearly nonnegative and self-adjoint, and trace-class, since $J$ and $J^\ast$ are Hilbert--Schmidt. Therefore, there exists an ONS $\{v_i\mid i\in I\}$ in $L^2_\Q$ and eigenvalues ${\mu}_i>0$, $i\in I$, for a countable index set $I$ with $|I|=\dim (\Ima J^\ast)$, such that $\sum_{i\in I} {\mu}_i<\infty$ and the spectral representation

\begin{equation}\label{SRJJa}
  J J^\ast = \sum_{i\in I} {\mu}_i \langle \cdot,v_i\rangle_\Q v_i
\end{equation}
holds. The summability of the eigenvalues ${\mu}_i$ implies that the convergence in \eqref{SRJJa} holds in the Hilbert--Schmidt norm sense. By the open mapping theorem, and since $\ker J J^\ast =\ker J^\ast$, we obtain that $JJ^\ast$ is invertible if and only if $\ker J^\ast=\{0\}$ and $\dim(L^2_\Q)<\infty$. It follows by inspection that $u_i = {\mu}_i^{-1/2} J^\ast v_i$ form an ONS in $\Hcal$ and that $J^\ast J u_i = {\mu}_i^{-1/2} J^\ast J J^\ast v_i = {\mu}_i u_i$. Then, since $\Hcal = \overline{\Ima J^\ast} \oplus \ker J $ and $\overline{\Ima J^\ast} = \overline{\spn\{u_i \mid i \in I\}}$, $J^\ast J$ has the spectral representation
\begin{equation}\label{SRJaJ}
  J^\ast J = \sum_{i\in I} {\mu}_i \langle \cdot,u_i\rangle_\Hcal u_i.
\end{equation}
As in \eqref{SRJJa}, the convergence in \eqref{SRJaJ} holds in the Hilbert--Schmidt norm sense. Furthermore, by analogous arguments as for $JJ^\ast$, we obtain that $J^\ast J$ is invertible if and only if $\ker J =\{0\}$ and $\dim(\Hcal)<\infty$. As a straightforward consequence of $\Hcal = \overline{\Ima J^\ast} \oplus \ker J$ and $L^2_\Q = \overline{\Ima J} \oplus \ker J^\ast$, we have the canonical expansions of $J^\ast$ and $J$ corresponding to \eqref{SRJJa} and \eqref{SRJaJ},
\begin{equation}\label{JJaexpansion}
 J^\ast = \sum_{i\in I}  {\mu}_i^{1/2}\langle \cdot,v_i\rangle_\Q u_i ,\quad  J = \sum_{i\in I} {\mu}_i^{1/2}\langle \cdot,u_i\rangle_\Hcal v_i.
\end{equation}

\begin{remark}
Note that \eqref{ass0} holds if and only if $J:\Hcal\to L^2_\Q$ is Hilbert--Schmidt. Indeed, \cite[Example 2.9]{ste_sco_12} shows a separable RKHS $\Hcal$ for which $J:\Hcal\to L^2_\Q$ is compact, but not Hilbert--Schmidt, and $\|\kappa\|_{2,\Q}=\infty$. That example also shows that $\kappa\notin\Hcal$ in general.
\end{remark}

\subsection{Proof of Lemma \ref{lemconv}}

Let $\{v_i \mid i \in I\}$ be the ONS in $L^2_\Q$ given in Section \ref{ssecHS}. Then $f_0 = \sum_{i \in I} \langle f_0, v_i \rangle_{2,\Q} v_i $. As $f_\lambda = J(J^\ast J +\lambda)^{-1} J^\ast f_0$, the spectral representation \eqref{SRJaJ} of $J^\ast J$ and  the canonical expansions \eqref{JJaexpansion} of $J^\ast$ and $J$ give $f_\lambda = \sum_{i \in I} \frac{{\mu}_i}{{\mu}_i + \lambda} \langle f_0, v_i \rangle_{2,\Q} v_i$. Hence,
\[\|f_0 - f_\lambda\|^2_{2,\Q} = \left\|\sum_{i \in I} \frac{\lambda}{{\mu}_i + \lambda} \langle f_0, v_i \rangle_{2,\Q} v_i \right\|^2_{2,\Q} = \sum_{i \in I} (\frac{\lambda}{{\mu}_i + \lambda})^2 \langle f_0, v_i \rangle^2_{2,\Q}. \]
The result follows from the dominated convergence theorem.

\subsection{Proof of Theorem \ref{thmLT}}

For simplicity, we assume that the sampling measure $\widetilde\Q=\Q$, that is, $w=1$, and omit the tildes. The extension to the general case is straightforward, using \eqref{diagHtilde} and \eqref{eqFXdef}.

We write
\begin{align*}
      f_{\bm X} - f_\lambda  &= (J^\ast_{\bm X}J_{\bm X} + \lambda )^{-1}J_{\bm X}^\ast  f - (J^\ast J + \lambda )^{-1}J^\ast f \\
    & = (J^\ast_{\bm X}J_{\bm X} + \lambda )^{-1} (J_{\bm X}^\ast f - J^\ast f) - ( (J^\ast J + \lambda)^{-1}-(J^\ast_{\bm X}J_{\bm X} + \lambda)^{-1} )J^\ast f .
\end{align*}
Combining this with the elementary factorization
\begin{equation}\label{trick}
  (J^\ast J + \lambda)^{-1} -(J^\ast_{\bm X}J_{\bm X} + \lambda)^{-1}  =  (J^\ast_{\bm X}J_{\bm X} + \lambda)^{-1}(  J^\ast_{\bm X}J_{\bm X}-J^\ast J)(J^\ast J + \lambda)^{-1} ,
\end{equation}
we obtain
\begin{equation}\label{keyeq}
f_{\bm X} - f_\lambda  = (J^\ast_{\bm X}J_{\bm X} + \lambda)^{-1} \left(J_{\bm X}^\ast  f - J^\ast f - ( J^\ast_{\bm X}J_{\bm X}-J^\ast J) f_\lambda  \right)  = (J^\ast_{\bm X}J_{\bm X} + \lambda)^{-1} \frac{1}{n} \sum_{i=1}^n \xi_i,
\end{equation}
where $\xi_i= (f(X^{(i)}) - f_\lambda(X^{(i)})) k_{X^{(i)}} - J^\ast (f - f_\lambda)$ are i.i.d.\ $\Hcal$-valued random variables with zero mean. Moreover, as
\begin{equation} \label{normxi2}
  \begin{aligned}
    \|\xi_i\|^2_{\Hcal} &= (f(X^{(i)}) - f_\lambda(X^{(i)}))^2 \kappa(X^{(i)})^2  + \int_{E^2} (f(x) - f_\lambda(x))(f(y) - f_\lambda(y)) k(x,y) \Q(dx) \Q(dy)\\
    &\quad-2 \int_E (f(X^{(i)})-f_\lambda(X^{(i)}))(f(y) - f_\lambda(y)) k(X^{(i)},y) \Q(dy),
\end{aligned}
\end{equation}
we infer that
\begin{equation}\label{xi_expec}
\E[\|\xi_i \|^2_\Hcal] =  \|(f-f_\lambda)\kappa \|_{2,\Q}^2 -  \|J^\ast (f-f_\lambda)\|^2_{\Hcal}\le  \|(f-f_\lambda)\kappa \|_{2,\Q}^2 \le 2\|f \kappa\|^2_{2,\Q} +  2\|f_\lambda\|_\Hcal^2\| \kappa\|^4_{4,\Q} <\infty ,
\end{equation}
where in the third inequality we used \eqref{eqfundamentalpnew}.

Hence both the law of large numbers in \eqref{Hlln} and the central limit theorem in \eqref{Hclt} apply:
\begin{equation}\label{keyclt1}
    \frac{1}{n} \sum_{i=1}^n \xi_i \xrightarrow{a.s.} 0, \quad  \frac{1}{\sqrt{n}} \sum_{i=1}^n \xi_i \xrightarrow{d} \Ncal(0, C_\xi),
\end{equation}
where $C_\xi$ is the covariance operator of $\xi$, which is given by
\begin{equation}\label{cov_xi}
    \langle C_\xi h, h  \rangle_{\Hcal} =  \| (f-f_\lambda) Jh\|^2_{2,\Q} -  \langle f-f_\lambda , Jh\rangle^2_{2,\Q}, \quad h \in \Hcal.
\end{equation}

From \eqref{keyeq}, \eqref{keyclt1} and Lemma~\ref{lemM1>0} below, the continuous mapping theorem gives $f_{\bm X} \xrightarrow{a.s.}  f_\lambda $, and Slutsky's lemma gives $\sqrt{n}(f_{\bm X} - f_\lambda) \xrightarrow{d} \Ncal(0, Q)$ for the covariance operator $Q=(J^\ast J+\lambda)^{-1}C_\xi (J^\ast J+\lambda)^{-1}$. Using \eqref{cov_xi}, we infer
\begin{align*}
    \langle Q h ,h\rangle_\Hcal& =  \| (f-f_\lambda) J(J^\ast J+\lambda)^{-1}h\|^2_{2,\Q}   -  \langle f-f_\lambda , J(J^\ast J+\lambda)^{-1}h\rangle^2_{2,\Q} \\
    & = \mathbb{V}_\Q[(f  - f_{\lambda} ) (J^\ast J+\lambda)^{-1} h] ,
\end{align*}
as claimed.

\begin{lemma}\label{lemM1>0}
We have $(J^\ast_{\bm X }J_{\bm X }+\lambda)^{-1} \xrightarrow{a.s.} (J^\ast J+\lambda)^{-1}$, as $n \to \infty$.
\end{lemma}

\begin{proof}[Proof of Lemma \ref{lemM1>0}]
Equation \eqref{trick} implies $\|(J^\ast J + \lambda)^{-1} -(J^\ast_{\bm X}J_{\bm X} + \lambda)^{-1} \| \le    \lambda^{-2}\| J^\ast_{\bm X}J_{\bm X}-J^\ast J\|$. Hence it is enough to prove that
\begin{equation}\label{lemM1>0c1}
 J^\ast_{\bm X }J_{\bm X }   \xrightarrow{a.s.}  J^\ast J .
\end{equation}
Thereto, we decompose
\begin{equation}\label{decJaJJaJ}
 J^\ast_{\bm X}J_{\bm X} - J^\ast J = \frac{1}{n} \sum_{i=1}^n \Xi_i,
\end{equation}
where $\Xi_i = \langle \cdot, k_{X^{(i)}} \rangle_\Hcal k_{X^{(i)}} - \int_E \langle \cdot, k_{x} \rangle_\Hcal k_x \Q(dx)$ are i.i.d.\ random Hilbert--Schmidt operators with zero mean. Straightforward calculations show that
\begin{equation}\label{decJaJJaJ2}
    \|\Xi_i\|^2_2 = \kappa(X^{(i)})^4 + \int_{E^2} k(x,y)^2 \Q(dx) \Q(dy) - 2 \int_E k(x,X^{(i)})^2 \Q(dx).
\end{equation}
It follows that
\begin{equation}\label{decJaJJaJ3}
  \E_\Q[\|\Xi_i\|^2_2] = \|\kappa\|^4_{4,\Q} - \int_{E^2} k(x,y)^2 \Q(dx)\Q(dy)<\infty.
\end{equation}
Hence the law of large numbers in \eqref{Hlln} applies and \eqref{lemM1>0c1} follows.
\end{proof}

\subsection{Proof of Theorem \ref{thmcinew}}

As in the proof of Theorem~\ref{thmLT}, we assume that the sampling measure $\widetilde\Q=\Q$, that is, $w=1$, and omit the tildes. The extension to the general case is straightforward, using \eqref{diagHtilde} and \eqref{eqFXdef}.

From \eqref{keyeq}, we infer $\|f_{\bm X} - f_\lambda \|_\Hcal  \le \frac{1}{\lambda} \|\frac{1}{n} \sum_{i=1}^n \xi_i\|_\Hcal$, and hence $\bm Q [\|f_{\bm X} - f_\lambda \|_\Hcal \ge \tau ] \le \bm Q\left[ \frac{1}{\lambda} \|\frac{1}{n} \sum_{i=1}^n \xi_i\|_\Hcal \ge  \tau \right]$. From \eqref{normxi2}, we infer
\[  \|\xi_i\|_\Hcal \le 2 \|(f-f_\lambda)\kappa\|_{\infty, \Q} \le 2\|f \kappa\|_{\infty, \Q} + 2\|f_\lambda\|_\Hcal \|\kappa\|^2_{\infty, \Q} < \infty,\]
where in the second inequality we used \eqref{eqfundamentalpnew}. Hence the Hoeffding inequality in \eqref{hoeffdingeq} applies, so that
\begin{equation}\label{ci_for_xi}
\bm Q \left[\left\|\frac{1}{n} \sum_{i=1}^n \xi_i\right\|_\Hcal \ge  \tau \right] \le  2\e^{-\frac{\tau^2 n}{8 \|(f-f_\lambda)\kappa\|^2_{\infty, \Q}}},\quad \tau>0,
\end{equation}
which implies \eqref{eqCInew}.

\subsection{Proof of Lemma~\ref{lemoptwnew}}
By definition we have $\widetilde\kappa=\kappa/\sqrt{w}$. From~\eqref{diagHtilde} we obtain $\|\widetilde\kappa\|_{\infty,\Q}\ge\|\widetilde\kappa\|_{2,\widetilde\Q}=\|\kappa\|_{2,\Q}$, with equality if and only if $\widetilde\kappa$ is constant $\Q$-a.s. This proves the lemma.

\subsection{Proof of Lemma \ref{lemGEKuniversalnew}}
Denote by $\Hcal_G$ the RKHS corresponding to the Gaussian kernel $k_G(x,y)=\e^{-\alpha\|x-y\|^2}$. It is well known that $\Hcal_G$ is densely embedded in $L^2_\Q$, see \cite[Proposition~8]{sri_fuk_lan_10}. Denote by $\Hcal_E$ the RKHS corresponding to the exponentiated kernel $k_E(x,y)=\e^{\beta x^\top y}$. As $k(x,y)=k_E(x,y)  k_G(x,y)$, and as $\Hcal_E$ contains the constant function, $1=k_E(\cdot,0)\in \Hcal_E$, we conclude from \cite[Theorem 5.16]{pau_rag_16} that $\Hcal_G\subset \Hcal$. This proves the lemma.

\section{Finite-dimensional target space}\label{secdimL2n}

We discuss the case where the target space $L^2_\Q$ from Section~\ref{secapprox} is finite-dimensional. This is of independent interest and provides the basis for computing the sample estimator without sorting.

Assume that $\Q=\frac{1}{n}\sum_{i=1}^n  \delta_{x_i}$, where $\delta_x$ denotes the Dirac point measure at $x$, for a sample of (not necessarily distinct) points $x_1,\dots,x_n\in E$, for some $n\in\N$. Then property~\eqref{ass0} holds, for any measurable kernel $k:E\times E\to\R$.

Note that $\bar n=\dim L^2_\Q\le n$, with equality if and only if $x_i\neq x_j$ for all $i\neq j$.  We discuss this in more detail now. Let $\bar x_1,\dots,\bar x_{\bar n}$ be the distinct points in $E$ such that $\{\bar x_1,\dots,\bar x_{\bar n}\}=\{x_1,\dots,x_n\}$. Define the index sets $I_j =\{ i\mid x_i=\bar x_j\}$, $j=1,\dots,\bar n$, so that
 \begin{equation}\label{eqvalt}
   \Q = \frac{1}{n}\sum_{j=1}^{\bar n}  |I_j| \delta_{\bar x_j}.
 \end{equation}
Then \eqref{eqJastgnew} reads $J^\ast g  = \frac{1}{n} \sum_{j=1}^{\bar n}   k(\cdot,\bar x_j) |I_j| g(\bar x_j) $, so that
\begin{equation}\label{eqJxJxa}
 JJ^\ast g (\bar x_i) =  \frac{1}{n} \sum_{j=1}^{\bar n}  k(\bar x_i,\bar x_j) |I_j| g(\bar x_j),\quad i=1,\dots,\bar n, \quad g\in L^2_\Q.
\end{equation}

We denote by $V_n$ the space $\R^n$ endowed with the scaled Euclidean scalar product $\langle y,z\rangle_n =\frac{1}{n} y^\top z$. We define the linear operator $S:\Hcal\to V_n$ by
\begin{equation}\label{Sdef}
  Sh = (h(x_1),\dots,h(x_n))^\top,\quad h\in\Hcal.
\end{equation}
Its adjoint is given by  $S^\ast y = \frac{1}{n}\sum_{j=1}^n k(\cdot,x_j) y_j$, so that
\begin{equation}\label{eqSSast}
  (SS^\ast y)_i = \frac{1}{n}\sum_{j=1}^n k(x_i,x_j)  y_j,\quad i=1,\dots,n,\quad y\in V_n.
\end{equation}
We define the linear operator $P:V_n\to L^2_\Q$ by $P y(\bar x_j)=\frac{1}{|I_j|}\sum_{i\in I_j} y_i$, $j=1,\dots,\bar n$, $y\in V_n$. Combining this with \eqref{eqvalt} we obtain $\langle P y,g\rangle_\Q = \frac{1}{n}\sum_{j=1}^{\bar n} |I_j| P y(\bar x_j) g(\bar x_j)=\frac{1}{n}\sum_{i=1}^n y_i g(x_i)$, for any $g\in L^2_\Q$. It follows that the adjoint of $P$ is given by $P^\ast g = (g(x_1),\dots,g(x_n))^\top$. In view of \eqref{Sdef}, we see that
\begin{equation}\label{ImaSsImaPa}
  \Ima S\subseteq \Ima P^\ast,
\end{equation}
and $PP^\ast$ equals the identity operator on $L^2_\Q$,
\begin{equation}\label{eqidPPa}
   P P^\ast g  = g,\quad g\in L^2_\Q.
\end{equation}
We claim that $J=PS$, that is, the following diagram commutes:
\begin{equation}\label{eqJPS}
 \begin{tikzcd}
 \& V_n \arrow[d,"P"] \\
  \Hcal \arrow[r,"J"]\arrow[ur,"S"] \& L^2_{\Q}
\end{tikzcd}
\end{equation}
Indeed, for any $h\in\Hcal$, we have $PSh (\bar x_j) = \frac{1}{|I_j|} \sum_{i\in I_j} h(x_i) = h(\bar x_j)$, which proves \eqref{eqJPS}.

Combining \eqref{ImaSsImaPa}--\eqref{eqJPS}, we obtain
\begin{equation}\label{eqkerJkerSnew}
  \ker J = \ker S
\end{equation}
and $P^\ast(J J^\ast+\lambda) = (S S^\ast+\lambda)P^\ast$. This is a useful result for computing the sample estimators below. Indeed, as $\lambda>0$ , we have that $g_\lambda$ in \eqref{eqKRRalt} is uniquely determined by the lifted equation
\begin{equation}\label{eqliftednn}
  (SS^\ast+\lambda) P^\ast g_\lambda =  P^\ast  f.
\end{equation}
In order to compute $f_\lambda=J^\ast g_\lambda = S^\ast P^\ast g_\lambda$, we can thus solve the $n\times n$-dimensional linear problem \eqref{eqliftednn}, with $P^\ast f\in V_n$ given, instead of the corresponding $\bar n\times \bar n$-dimensional linear problem \eqref{eqKRRalt}. This fact allows for faster implementation of the sample estimation, as the test of whether $\bar n<n$ for a given sample $x_1,\dots,x_n$ is not needed, see Lemma~\ref{lemcompnew} below.

\subsection{Computation without sorting}\label{sseccompwos}

As an application of the above, we now discuss how to compute the sample estimator in \eqref{eqFXdef} without sorting the sample $\bm X$. Thereto, we fix the orthogonal basis $\{e_1,\dots,e_{n}\}$ of $V_n$ given by $e_{i,j}= \delta_{ij}$, so that $\langle e_i,e_j\rangle_{n}=\frac{1}{n}\delta_{ij}$, for $1\le i,j\le n$. We denote by $\overline{\bm f}=(\widetilde f(X^{(1)}),\dots,\widetilde f(X^{(n)}))^\top$ and define the positive semidefinite $n\times n$-matrix $\overline{\bm K}$ by $\overline{\bm K}_{ij}=\widetilde k(X^{(i)},X^{(j)})$. From \eqref{eqSSast} we see that $\frac{1}{n}\overline{\bm K}$ is the matrix representation of $\widetilde S \widetilde S^\ast:V_n\to V_n$. Summarizing, we arrive at the following alternative to Lemma~\ref{lemcompnewX}.

\begin{lemma}\label{lemcompnew}
The unique solution $\overline{\bm g}\in\R^n$ to
\begin{equation}\label{KLS}
  \textstyle(\frac{1}{n}\overline{\bm K}+\lambda) \overline{\bm g} = \overline{\bm f},
\end{equation}
gives $ f_{\bm X}=\frac{1}{n}\sum_{i=1}^n  k(\cdot,X^{(i)})\frac{\overline{\bm g}_i}{\sqrt{w( X^{(i)})}}$. Moreover, the solutions of \eqref{KLS'} and \eqref{KLS} are related by $\overline{\bm g}_i=|I_j|^{-1/2}{\bm g}_j$ for all $i\in I_j$, $j=1,\dots,\bar n$.
\end{lemma}

\begin{remark}
If $X^{(i)}\neq X^{(j)}$ for all $i\neq j$ (that is, if $\bar n=n$), then $\overline{\bm K}=\bm K$, $\overline{\bm f}=\bm f$, and Lemmas~\ref{lemcompnewX} and \ref{lemcompnew} coincide. Otherwise they provide different computational schemes.
\end{remark}

\section{Finite-dimensional RKHS}\label{secdimH}

We discuss the case where the RKHS $\Hcal$ from Section~\ref{secapprox} is finite-dimensional in more detail. In particular, we then extend some of our results to the case without regularization, $\lambda=0$.

Let $\{\phi_1,\dots,\phi_m\}$ be a set of linearly independent measurable functions on $E$ with $\|\phi_i\|_{2,\Q}<\infty$, $i=1,\dots,m$, for some $m\in\N$. Denote the \emph{feature map} $\phi=(\phi_1,\dots,\phi_m)^\top:E\to\R^m$ and define the measurable kernel $k:E\times E\to\R$ by $k(x,y)=\phi(x)^\top\phi(y)$. It follows by inspection that \eqref{ass0} holds and $\{\phi_1,\dots,\phi_m\}$ is an ONB of $\Hcal$, which is in line with the Lemma~\ref{thmmercernew}\ref{thmmercernew1}. Hence any function $h\in\Hcal$ can be represented by the coordinate vector $\bm h=\langle h,\phi\rangle_\Hcal\in\R^m$, $h=\phi^\top \bm h$. The operator $J^\ast:L^2_\Q\to\Hcal$ is of the form $J^\ast g = \phi^\top\langle\phi,g\rangle_\Q$. Hence $J^\ast J:\Hcal\to\Hcal$ satisfies $J^\ast J \phi^\top = \phi^\top \langle \phi,\phi^\top\rangle_\Q$, and can thus be represented by the $m\times m$-Gram matrix $\langle \phi,\phi^\top\rangle_\Q$. That is, $J^\ast J h=J^\ast J\phi^\top\bm h = \phi^\top \langle \phi,\phi^\top\rangle_\Q \bm h$, for $h\in\Hcal$.

We henceforth assume that $\ker J=\{0\}$, so that $J^\ast J:\Hcal\to\Hcal$ is invertible, by Section~\ref{ssecHS}. This is equivalent to $\{ J\phi_1,\dots,J\phi_m\}$ being a linearly independent set in $L^2_\Q$. We transform it into an ONS. Consider the spectral decomposition $\langle \phi,\phi^\top\rangle_\Q = S D S^\top$ with orthogonal matrix $S$ and diagonal matrix $D$ with $D_{ii}>0$. Define the functions $\psi_i\in\Hcal$ by $\psi^\top =(\psi_1 ,\dots,\psi_m ) =  \phi^\top S D^{-1/2} $ . Then $\langle  \psi,\psi^\top\rangle_\Q =D^{-1/2} S^\top \langle \phi,\phi^\top\rangle_\Q  S D^{-1/2} = I_m$, so that $\{J\psi_1,\dots,J\psi_m\}$ is an ONS in $L^2_\Q$. Moreover, we have $J^\ast J\psi^\top = J^\ast J\phi^\top S D^{-1/2}  =   \phi^\top \langle \phi,\phi^\top\rangle_\Q S D^{-1/2}  = \psi^\top D$,  so that $v_i=J\psi_i$ are the eigenvectors of $JJ^\ast$ with eigenvalues
\begin{equation}\label{evJfindim}
  {\mu}_i=D_{ii}>0 ,\quad i=1,\dots,m,
\end{equation}
and the spectral decomposition~\eqref{SRJJa} holds with index set $I=\{1,\dots,m\}$. The corresponding ONB of $\Hcal$ in the spectral decomposition~\eqref{SRJaJ} is given by $(u_1,\dots,u_m) =  J^\ast J\psi^\top D^{-1/2}  = \psi^\top D^{1/2} =\phi^\top S$. Note that we can express the kernel directly in terms of the rotated feature map $u$, $k(x,y)= u(x)^\top u(y)$, in line with Lemma~\ref{thmmercernew}\ref{thmmercernew1}.

\subsection{Approximation without regularization}

As $J^\ast J:\Hcal\to\Hcal$ is invertible, it follows that problem \eqref{KRR} always has a unique solution for $\lambda=0$, which obviously coincides with the projection $f_0 = (J^\ast J)^{-1} J^\ast f$.

\subsection{Sample estimation without regularization}

As in Section~\ref{secFSE}, we let $n\in\N$ and $\bm X=(X^{(1)},\dots,X^{(n)})$ be a sample of i.i.d.\ $E$-valued random variables with $X^{(i)}\sim\widetilde\Q$. We henceforth assume that $\lambda=0$, and hence we have to address the case where $\widetilde J_{\bm X}^\ast \widetilde J_{\bm X}$ is not invertible on $\widetilde\Hcal$. In this case, we shall denote by ``$(\widetilde J_{\bm X}^\ast \widetilde J_{\bm X})^{-1}$'' any linear operator on $\widetilde\Hcal$ that coincides with the inverse of $\widetilde J_{\bm X}^\ast \widetilde J_{\bm X}$ restricted to $\Ima \widetilde J_{\bm X}^\ast\subset\widetilde\Hcal$. As a consequence, $\widetilde f_{\bm X} =(\widetilde J_{\bm X}^\ast \widetilde J_{\bm X})^{-1} \widetilde J_{\bm X}^\ast f$ is always well defined and solves problem \eqref{KRR} with $\lambda=0$ and $\Q$ replaced by $\widetilde\Q_{\bm X}$.

We first show that our limit theorems carry over. The proof is given in Section~\ref{secproofthmLT0}.
\begin{theorem}\label{thmLT0}
  Theorem~\ref{thmLT} literally applies for $\lambda=0$, and so does Remark \ref{remCLT} (but not Remark~\ref{remthmLT}).
\end{theorem}

We denote by $\underline\mu=\min_{i\in I}\mu_i >0$ the minimal eigenvalue of $J^\ast J$, see \eqref{evJfindim}. The finite sample guarantee in Theorem~\ref{thmcinew} is modified as follows. The proof is given in Section~\ref{secproofthmcinew0}.

\begin{theorem}\label{thmcinew0}
For any $\eta\in (0,1]$, we have
\begin{equation}\label{eqCInew0}
  \|f_{\bm X} - f_0\|_\Hcal  <  \frac{2\sqrt{2 \log(4/\eta)}\|(1/w)(f-f_0)\kappa\|_{\infty,\Q}}{(1-C(\eta)/\sqrt{n})\underline\mu\sqrt{n}}
\end{equation}
with sampling probability $\bm Q$ of at least $1-\eta$, where $C(\eta)=  2\sqrt{\log(4/\eta)} \underline\mu^{-1} \|\widetilde\kappa\|^2_{\infty, \Q} $, for all $n>C(\eta)^2$.
\end{theorem}

Theorem~\ref{thmcinew0} is similar to \cite[Theorem 2.1(iii)]{coh_mig_17}, but in contrast extends to unbounded $f$ under assumptions \eqref{newasstildef} and \eqref{newasstildek}, and provides a learning rate $O((\frac{\log n}{n})^{1/2})$ for the sample error (set $\eta=n^{-r}$, for some $r>0$).

\subsection{Computation }

We now revisit Section~\ref{seccompnew} for the case of a finite-dimensional RKHS $\Hcal$. Note that $\widetilde\phi_j=\phi_j/\sqrt{w}$ form an ONB of $\widetilde\Hcal$. We define the $\bar n\times m$-matrix ${\bm V}$ by ${\bm V}_{ij} = |I_i|^{1/2}\widetilde\phi_j(\bar X^{(i)})$, so that ${\bm K}={\bm V} {\bm V}^\top$, which is given in Section~\ref{seccompnew}. Then ${\bm V}$ is the matrix representation of $\widetilde J_{\bm X}:\widetilde\Hcal\to L^2_{\widetilde\Q_{\bm X}}$, also called the \emph{design matrix}, and $\frac{1}{n}{\bm V}^\top$ is the matrix representation of $\widetilde J_{\bm X}^\ast:L^2_{\widetilde\Q_{\bm X}}\to\widetilde\Hcal$.\footnote{The matrix transpose ${\bm V}^\top$ is scaled by $\frac{1}{n}$ because the orthogonal basis $\{\psi_1,\dots,\psi_{\bar n}\}$ of $L^2_{\widetilde\Q_{\bm X}}$ is not normalized.}  Note that $k$ is tractable if and only if $\E_\Q[ \phi(X)\mid\Fcal_t]$ is given in closed form for all $t$. We arrive at the following result, which corresponds to Lemma~\ref{lemcompnewX} and which holds for any $\lambda\ge 0$. In case where $\lambda=0$, we assume that $\ker \widetilde J_{\bm X}=\{0\}$, so that $\widetilde J_{\bm X}^\ast \widetilde J_{\bm X}$ is invertible.

\begin{lemma}\label{lemcompnewFDnew}
The unique solution $ \bm h\in\R^m$ to
\begin{equation}\label{VtVLSnew}
 \textstyle(\frac{1}{n} {\bm V}^\top {\bm V} +\lambda)  \bm h = \frac{1}{n} {\bm V}^\top {\bm f},
\end{equation}
gives $ f_{\bm X} = \phi^\top \bm h$. The sample version of problem \eqref{KRR},
\begin{equation}\label{KRRempnew}
  \min_{\bm h\in\R^m} ( \frac{1}{n} \|{\bm V} \bm h - {\bm f}\|^2  + \lambda\|\bm h\|^2),
\end{equation}
has a unique solution $\bm h\in \R^m$, which coincides with the solution to \eqref{VtVLSnew}. If, moreover, the kernel $k$ is tractable then
\begin{equation}\label{hatVcf2}
V_{\bm X,t}=\E_\Q[ \phi(X)\mid\Fcal_t]^\top \bm h,\quad t=0,\dots,T,
\end{equation}
is given in closed form.
\end{lemma}

The least-squares problem~\eqref{KRRempnew} can be efficiently solved using stochastic gradient methods such as the randomized extended Kaczmarz algorithm in \cite{zou_fre_13, fil_gla_nak_sta_19}.

\subsection{Computation without sorting}

Following up on Section~\ref{sseccompwos}, we define the $n\times m$-matrix $\overline{\bm V}$ by $\overline{\bm V}_{ij} = \widetilde\phi_j(X^{(i)})$, so that $\overline{\bm K}=\overline{\bm V} \overline{\bm V}^\top$. Note that $\overline{\bm V}$ is the matrix representation of $\widetilde S:\widetilde\Hcal\to V_n$ in \eqref{Sdef}, and $\frac{1}{n}\overline{\bm V}^\top$ is the matrix representation of $\widetilde S^\ast:V_n\to\widetilde\Hcal$.\footnote{The matrix transpose $\overline{\bm V}^\top$ is scaled by $\frac{1}{n}$ because the orthogonal basis $\{e_1,\dots,e_{n}\}$ of $V_n$ is not normalized.}
From \eqref{eqkerJkerSnew} we thus infer that $\ker \overline{\bm V} = \ker \widetilde J_{\bm X}$. As a consequence, or by direct verification, we further obtain $\overline{\bm V}^\top \overline{\bm V}=\bm V^\top \bm V$, $\overline{\bm V}^\top \overline{\bm f} = \bm V^\top \bm f$, and $\|\overline{\bm V} \bm h - \overline{\bm f}\|  = \|\bm V \bm h - \bm f\|$. Summarizing, we thus infer that Lemma~\ref{lemcompnewFDnew} literally applies to $\overline{\bm V}$ and $\overline{\bm f}$ in lieu of $\bm V$ and $\bm f$.

\subsection{Proof of Theorem \ref{thmLT0}}\label{secproofthmLT0}

As in the proof of Theorem \ref{thmLT}, we assume for simplicity that the sampling measure $\widetilde\Q=\Q$, that is, $w=1$, so that we can omit the tildes.

We fix $\delta\in [0,1)$, and define the sampling event $\Scal_\delta =\{ \|J^\ast_{\bm X} J_{\bm X} - J^\ast J\|_2\le \delta / \| (J^\ast J)^{-1}\|\}\subseteq \bm E $. The following lemma collects some properties of $\Scal_\delta$.

\begin{lemma}\label{lemBdelta}
\begin{enumerate}
  \item\label{lemBdelta1} On $\Scal_\delta$, the operator $J_{\bm X}^\ast J_{\bm X} :\Hcal\to\Hcal$ is invertible and
\begin{equation}\label{Bdeltamux}
  \|(J^\ast_{\bm X }J_{\bm X })^{-1}\|
  \le \frac{\|(J^\ast J )^{-1}\|}{1-\delta}.
\end{equation}
\item\label{lemBdelta2} The sampling probability of $\Scal_\delta$ is bounded below by
 \begin{equation}\label{PaOdbound}
  {\bm Q}[ \Scal_\delta]\ge 1 -2\e^{-\frac{\delta^2   n}{4 \|\kappa\|_{\infty,\Q}^4 \|(J^\ast J )^{-1}\|^2}}.
\end{equation}
\end{enumerate}
\end{lemma}

\begin{proof}[Proof of Lemma \ref{lemBdelta}]
\ref{lemBdelta1}: We write $J^\ast_{\bm X} J_{\bm X}  = J^\ast J(J^\ast J )^{-1} J^\ast_{\bm X} J_{\bm X} $, so that $J^\ast_{\bm X}J_{\bm X} $ is invertible if and only if $(J^\ast J )^{-1}J^\ast_{\bm X} J_{\bm X} $ is invertible. If $\| (J^\ast J )^{-1}\|\| J^\ast J-J^\ast_{\bm X}J_{\bm X}\|_2\le  \delta$, then $\|1 - (J^\ast J  )^{-1}J^\ast_{\bm X}J_{\bm X} \| \le \delta$, which proves the invertibility of $(J^\ast J )^{-1}J^\ast_{\bm X}J_{\bm X} $, and hence of $J^\ast_{\bm X}J_{\bm X}$. Furthermore, using Neumann series of $1 - (J^\ast J)^{-1} J^\ast_{\bm X}J_{\bm X} $ we obtain \eqref{Bdeltamux}.

\ref{lemBdelta2}: We decompose $J^\ast_{\bm X}J_{\bm X} - J^\ast J$ as in \eqref{decJaJJaJ}. From \eqref{decJaJJaJ2} we infer that $\|\Xi_i\| \le \sqrt{2} \|\kappa\|^2_{\infty, \Q} <\infty$. Consequently, the Hoeffding inequality \eqref{hoeffdingeq} applies and we obtain
\begin{equation}\label{ci_A}
\bm Q[\|J^\ast_{\bm X}J_{\bm X}-J^\ast J\|_2 \ge   \tau]  \le 2 \e^{-\frac{\tau^2 n}{4\|\kappa \|^4_{\infty, \Q}}},
\end{equation}
which again is equivalent to \eqref{PaOdbound}.
\end{proof}

In view of Lemma~\ref{lemBdelta}\ref{lemBdelta1}, it now follows by inspection that \eqref{trick} and \eqref{keyeq} hold on $\Scal_\delta$ for $\lambda=0$. We thus obtain the global identity
\begin{equation}\label{keyeql0}
f_{\bm X} - f_0  = \Delta_{\bm X} + (J^\ast_{\bm X}J_{\bm X})^{-1} \frac{1}{n} \sum_{i=1}^n \xi_i,
\end{equation}
where the $\Hcal$-valued random variable $\Delta_{\bm X} = f_{\bm X} - f_0 - (J^\ast_{\bm X}J_{\bm X})^{-1} \frac{1}{n} \sum_{i=1}^n \xi_i$ satisfies $\Delta_{\bm X}=0$ on $\Scal_\delta$. In view of \eqref{PaOdbound} and the Borel--Cantelli lemma, we thus have $\sqrt{n}\Delta_{\bm X} \xrightarrow{a.s.} 0$, as $n \to \infty$.

Note that \eqref{normxi2}--\eqref{cov_xi} clearly hold with $\lambda=0$. Theorem~\ref{thmLT0} now follows as in the proof of Theorem \ref{thmLT} with $\lambda=0$, with \eqref{keyeq} replaced by \eqref{keyeql0}, and with Lemma~\ref{lemM1>0} replaced by the following lemma.

\begin{lemma}\label{lemM1}
We have $(J^\ast_{\bm X }J_{\bm X })^{-1} \xrightarrow{a.s.} (J^\ast J)^{-1}$, as $n \to \infty$.
\end{lemma}

\begin{proof}[Proof of Lemma~\ref{lemM1}]
Let $\tau>0$. We have
\begin{equation}\label{lemM1eq1}
    \bm Q[\|(J^\ast_{\bm X}J_{\bm X} )^{-1} - (J^\ast J )^{-1}\| \ge \tau] = \bm Q[\|(J^\ast_{\bm X}J_{\bm X} )^{-1} - (J^\ast J )^{-1}\| \ge \tau, \Scal_\delta]  + \bm Q[ \bm E \setminus \Scal_\delta].
\end{equation}
Using \eqref{trick} and \eqref{Bdeltamux}, we obtain on $\Scal_\delta$,
\[    \|(J^\ast_{\bm X}J_{\bm X} )^{-1} - (J^\ast J )^{-1}\| \le \frac{\|(J^\ast J )^{-1}\|^2}{1-\delta}\|J^\ast_{\bm X}J_{\bm X} -J^\ast J \|_2.\]
Combining this with \eqref{ci_A}, we obtain
\[  \bm Q[\|(J^\ast_{\bm X}J_{\bm X} )^{-1} - (J^\ast J )^{-1}\| \ge \tau, \Scal_\delta] \le \bm Q\left[\frac{\|(J^\ast J )^{-1}\|^2}{1-\delta}\|J^\ast_{\bm X}J_{\bm X}-J^\ast J\|_2 \ge \tau\right]  \le 2\e^{ \frac{-\tau^2(1-\delta)^2 n}{4 \|\kappa\|_{\infty, \Q}^4 \|(J^\ast J )^{-1}\|^4}}.\]
Combining this with \eqref{PaOdbound} and \eqref{lemM1eq1}, we infer that
\[     \bm Q[\|(J^\ast_{\bm X}J_{\bm X} )^{-1} - (J^\ast J )^{-1}\| \ge \tau] \le 2   \e^{ \frac{-\tau^2(1-\delta)^2 n}{4 \|\kappa\|_{\infty, \Q}^4 \|(J^\ast J )^{-1}\|^4}}  + 2   \e^{\frac{-\delta^2   n}{4 \|\kappa\|_{\infty,\Q}^4 \|(J^\ast J  )^{-1}\|^2}}.\]
As the right-hand side is summable over $n\ge 1$ for any $\tau >0$, the lemma follows from the Borel--Cantelli lemma.
\end{proof}

\subsection{Proof of Theorem \ref{thmcinew0}}\label{secproofthmcinew0}

As in the proof of Theorem~\ref{thmcinew}, we assume that the sampling measure $\widetilde\Q=\Q$, that is, $w=1$. The extension to the general case is straightforward, using \eqref{diagHtilde} and \eqref{eqFXdef}.

We let the sampling event $\Scal_\delta$ be as in Lemma~\ref{lemBdelta}, and let $\tau>0$. We have
\begin{equation}\label{JaJlaminvnormeq1}
    {\bm Q} [ \| f_{\bm X} - f_0\|_\Hcal  \ge  \tau ]  \le  {\bm Q} [ \| f_{\bm X} - f_0\|_\Hcal  \ge  \tau , \Scal_\delta] + {\bm Q} [\bm E\setminus\Scal_\delta].
\end{equation}
Using \eqref{keyeq} and \eqref{Bdeltamux}, we obtain on $\Scal_\delta$,
\[ \| f_{\bm X} - f_0\|_\Hcal \le \frac{\|(J^\ast J )^{-1}\|}{1-\delta}\|\frac{1}{n}\sum_{i=1}^n \xi_i\|_{\Hcal}.\]
Combining this with \eqref{ci_for_xi}, we obtain
\[ {\bm Q} [ \| f_{\bm X} - f_0\|_\Hcal  \ge  \tau , \Scal_\delta] \le \bm Q\left[\frac{\|(J^\ast J )^{-1}\|}{1-\delta}\left\|\frac{1}{n}\sum_{i=1}^n \xi_i\right\|_{\Hcal}\ge  \tau\right]  \le 2\e^{-\frac{\tau^2 (1-\delta)^2 n}{8 \|(f-f_0)\kappa\|^2_{\infty, \Q}\|(J^\ast J )^{-1}\|^2}}.\]
Combining this with \eqref{PaOdbound} and \eqref{JaJlaminvnormeq1}, we infer that
\[     {\bm Q} [ \| f_{\bm X} - f_0\|_\Hcal  \ge  \tau ] \le 2\e^{-\frac{\tau^2 (1-\delta)^2 n}{8 \|(f-f_0)\kappa\|^2_{\infty, \Q}\|(J^\ast J )^{-1}\|^2}} + 2   \e^{\frac{-\delta^2   n}{4 \|\kappa\|_{\infty,\Q}^4 \|(J^\ast J  )^{-1}\|^2}}.\]

Now we choose $\delta = \frac{\|\kappa\|_{\infty, \Q}^2 \tau }{\sqrt{2} \|(f-f_0)\kappa\|_{\infty, \Q}+\|\kappa\|_{\infty, \Q}^2 \tau}$, so that the two exponents on the right hand side match. Therefore, we obtain
\[ {\bm Q} [ \| f_{\bm X} - f_0\|_\Hcal  \ge  \tau ] \le 4\e^{-\frac{\delta^2 n}{4\|\kappa\|^4_{\infty, \Q}\|(J^\ast J )^{-1}\|^2}} = 4\e^{-\frac{\tau^2n}{4\|(J^\ast J )^{-1}\|^2(\sqrt{2}\|(f-f_0)\kappa\|_{\infty, \Q} + \|\kappa\|_{\infty,\Q}^2 \tau)^2}}.\]
Straightforward rewriting gives \eqref{eqCInew0}, where we use the fact that $\| (J^\ast J)^{-1}\|= \underline\mu^{-1}$, see \eqref{SRJaJ}.

\section{Comparison with regress-now}\label{appregnow}
The method we developed in this paper gives an estimation of the entire value process $V$. In practice one could be interested in the estimation of the portfolio value $V_t$ only at some fixed time $t$, e.g., $t=1$. In \cite{gla_yu_04} two least squares Monte Carlo approaches are described to deal with this problem in the context of American options pricing. The first approach, called ``regress-later'', consists in approximating the payoff function $f$ by means of a projection on a finite number of basis functions. The basis functions are chosen in a way that their conditional expectation at time $t=1$ is in closed form. Our method can be seen as a double extension of this approach, because it also covers the case where the number of basis functions could potentially be infinite, and gives closed form estimation of the portfolio value $V_t$ at any time $t$. The second approach, called ``regress-now'', consists in approximating $V_1$ by means of a projection on a finite number of basis functions that depend solely on the variable of interest, $x_1 \in E_1$.

We compare our approach, which corresponds to ``regress-later'' and which gives the estimator $V_{\bm X, 1}$ in \eqref{hatVcf1}, to its regress-now variant, whose estimator we denote by $V_{\bm X, 1}^{\text{now}}$. Thereto we briefly discuss how to construct $V_{\bm X, 1}^{\text{now}}$, and implement it, in the context of the three examples studied in Section \ref{secexamples}.

To construct $V_{\bm X, 1}^{\text{now}}$ only a few changes need to be carried out to the previous construction of $V_{\bm X, 1}$. We sample directly from $\Q$ which gives $\bm X = (X^{(1)}, \cdots, X^{(n)})$ and the vector $\bm f = (f(X^{(1)}), \cdots, f(X^{(n)}))^{\top}$. The expression of $V_{\bm X, 1}^{\text{now}}$ is given by \eqref{hatVcf1} for $t=1$ where the kernel $k$ is of the form \eqref{kernel} with $T$ replaced by $1$ so that its domain is $E_1\times E_1$. Instead of using the whole sample $\bm X$ to construct the matrix $\bm K$ in \eqref{KLS'} only the $(t=1)$-cross-section $\bm X_1 = (X^{(1)}_1,\cdots, X^{(n)}_1)$ is needed.\footnote{We omit $X_0$, as mentioned in footnote \ref{footnote_X0}.} Since the sampling measure $\Q$ is Gaussian, property \eqref{XiXjneq} holds for the sample $\bm X_1$ so that $\bar n=n$, $\bar X_1^{(j)}=X_1^{(j)}$ and $|I_j|=1$ for all $j=1,\dots,n$. The conditional kernel embeddings in \eqref{hatVcf1} boil down to $M_1(  X^{(j)})=k( X_1, X^{(j)}_1)$.

Table \ref{table_vnow_vlat} shows the normalized $L^2_{\Q}$-errors of $V_{\bm X, 1}^{\text{now}}$ and compares them to the respective values of $V_{\bm X, 1}$, which are taken from Table \ref{table1new}. We observe that for all three examples, our regress-later estimator performs better than the regress-now estimator. This finding is confirmed by Figure \ref{fig_regnow}, which corresponds to Figures~\ref{fig_minput}, \ref{fig_maxcall_highd} and \ref{fig_barrier}.

\begin{table}[h]
\centering
  \begin{tabular}{|l|c|c|}
\hline
Payoff &      regress-now &         regress-later  \\
\hline\hline
Min-put & 1.946& \bf{1.827}\\
\hline
Max-call &   2.606 & \bf{2.500}\\
\hline
Barrier reverse convertible   & 0.2806 & \bf{0.2506}\\
\hline
\end{tabular}
\caption{Normalized $L^2_\Q$-error $\|V_1- \widehat{V}_1\|_{2,\Q}/V_0$ in \% for $\widehat{V}_1 = V_{\bm X, 1}^{\text{now}}, V_{\bm X, 1}$ with $\gamma = 0$.}\label{table_vnow_vlat}
\end{table}

\begin{figure}[p]
    \centering 
\begin{subfigure}{0.45\textwidth}
  \includegraphics[width=\linewidth]{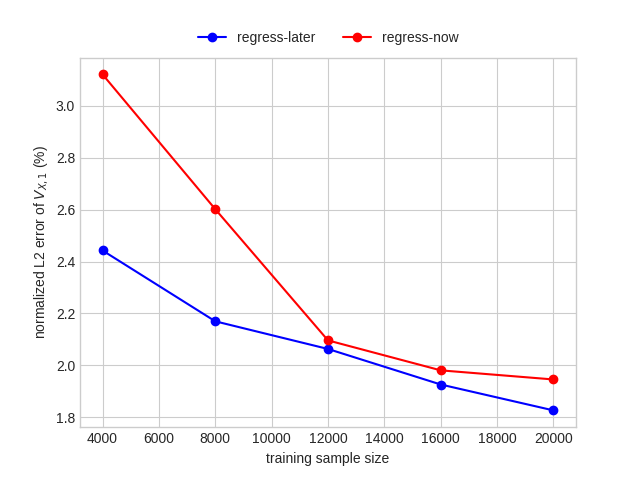}
  \caption{Min-put: normalized $L^2_{\Q}$-errors of $V_{\bm X, 1}$ and $V_{\bm X, 1}^{\text{now}}$ in \%}
  \label{error_regnow_minput}
\end{subfigure}\hfil 
\begin{subfigure}{0.45\textwidth}
  \includegraphics[width=\linewidth]{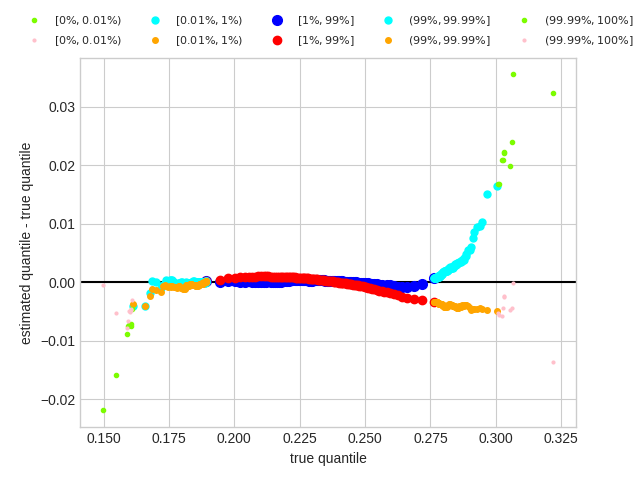}
  \caption{Min-put: detrended Q-Q plots of $V_{\bm X, 1}$ and $V_{\bm X, 1}^{\text{now}}$}
  \label{qqplot_regnow_minput}
\end{subfigure}
\medskip
\begin{subfigure}{0.45\textwidth}
  \includegraphics[width=\linewidth]{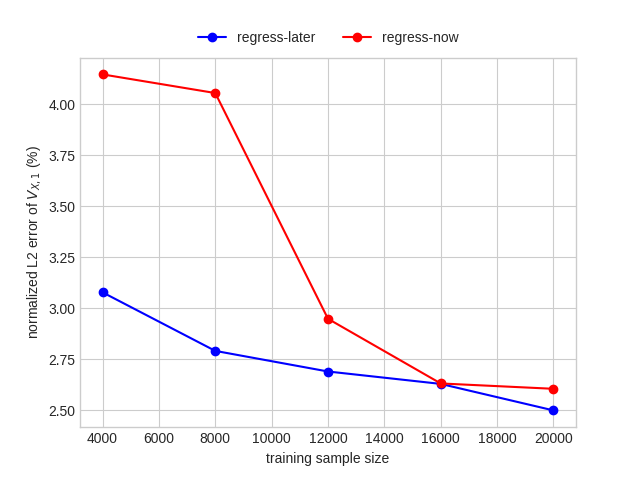}
  \caption{Max-call: normalized $L^2_{\Q}$-errors of $V_{\bm X, 1}$ and $V_{\bm X, 1}^{\text{now}}$ in \%}
  \label{error_regnow_maxcall}
\end{subfigure}\hfil 
\begin{subfigure}{0.45\textwidth}
  \includegraphics[width=\linewidth]{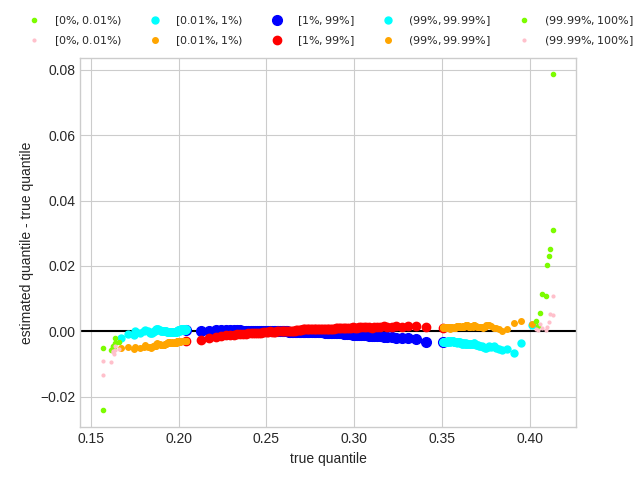}
  \caption{Max-call: detrended Q-Q plots of $V_{\bm X, 1}$ and $V_{\bm X, 1}^{\text{now}}$}
  \label{qqplot_regnow_maxcall}
\end{subfigure}
\medskip
\begin{subfigure}{0.45\textwidth}
  \includegraphics[width=\linewidth]{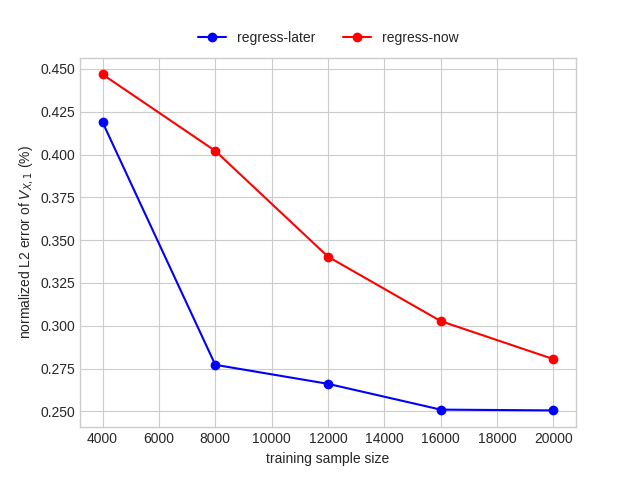}
  \caption{Barrier reverse convertible: normalized $L^2_{\Q}$-errors of $V_{\bm X, 1}$ and $V_{\bm X, 1}^{\text{now}}$ in \%}
  \label{error_regnow_barrier}
\end{subfigure}\hfil 
\begin{subfigure}{0.45\textwidth}
  \includegraphics[width=\linewidth]{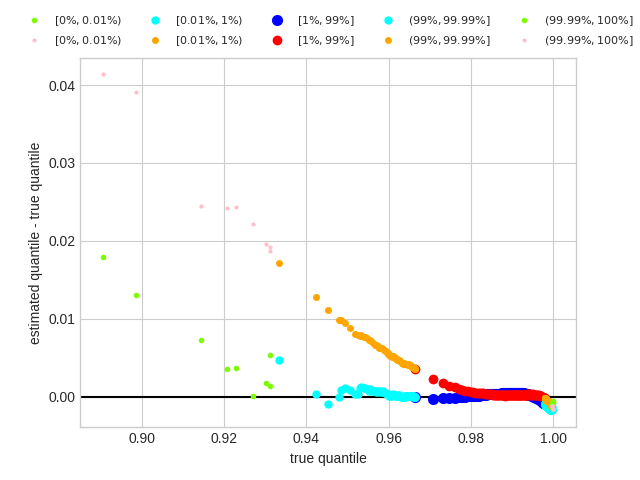}
  \caption{Barrier reverse convertible: detrended Q-Q plots of $V_{\bm X, 1}$ and $V_{\bm X, 1}^{\text{now}}$}
  \label{qqplot_regnow_barrier}
\end{subfigure}
\caption{Comparison of $V_{\bm X, 1}$ and $V_{\bm X, 1}^{\text{now}}$ with $\gamma = 0$. In the detrended Q-Q plots, the blue, cyan, and lawngreen (red, orange, and pink) dots are built using the regress-later (regress-now) estimator and the test data. $[0\%, 0.01\%)$ refer to the quantiles of levels $\{0.001\%,0.002\%, \cdots, 0.009\%\}$, $[0.01\%, 1\%)$ refer to the quantiles of levels $\{0.01\%,0.02\%, \cdots, 0.99\%\}$, $[1\%, 99\%]$ refer to the quantiles of levels $\{1\%,2\%, \cdots, 99\%\}$, $(99\%, 99.99\%]$ refer to the quantiles of levels $\{99.01\%,99.02\%, \cdots, 99.99\%\}$, and $(99.99\%, 100\%]$ refer to the quantiles of levels $\{99.991\%,99.992\%, \cdots, 100\%\}$. }
\label{fig_regnow}
\end{figure}

Table \ref{table_CP} shows the computing times for the full estimation of $V_{\bm X, 1}^{\text{now}}$ and $V_{\bm X, 1}$. Computation is performed on Skylake processors running at 2.3 GHz, using 14 cores and 100 GB of RAM. We see that our estimator requires less computing time than the regress-now estimator. Indeed, note that the dimension of the regression problem, which equals the training sample size $n=20{,}000$, is the same for both methods.

\begin{table}[h]
\centering
  \begin{tabular}{|l|c|c|}
  \hline
Payoff &      regress-now &         regress-later  \\
\hline\hline
Min-put & 3122 & \bf{2756} \\
\hline
Max-call &  3288  & \bf{2570} \\
\hline
Barrier reverse convertible   & 3683  & \bf{2756} \\
\hline
\end{tabular}
\caption{Computing times (in seconds) for estimating $V_{\bm X, 1}^{\text{now}}$ and $V_{\bm X, 1}$ with $\gamma = 0$, for sample size $n=20{,}000$.}\label{table_CP}
\end{table}

Thus, both in terms of accuracy in the estimation of $V_1$, measured by normalized $L^2_\Q$-error, and computing time, our regress-later estimator outperforms the regress-now estimator, and thus is better suited for portfolio valuation tasks. This might look surprising since the estimation of the entire value process $V$ by regress-later is a high dimension problem (path space dimension 12 for the min-put and max-call, and 36 for the barrier reverse convertible), whereas the direct estimation of $V_1$ by regress-now is a problem of much smaller dimension (state space dimension 6 for the min-put and max-call, and 3 for the barrier reverse convertible). A reason for the inferior performance of regress-now could be that the training data $\bm f$ represents noisy observations of the true values $ V_1(X^{(1)}_1), \cdots, V_1(X^{(n)}_1)$, which we cannot directly observe. This is in contrast to our regress-later approach, where $\bm f$ are the true values of the target function $f$. Our findings suggest that for portfolio valuation, it is more efficient to first estimate the payoff function $f$ on a high-dimensional domain and with non-noisy observations, rather than directly estimate the time $t=1$-value function $V_1$ on a low-dimensional domain but with noisy observations.

How about the risk measures? Tables \ref{table_var_nowVSlater} and \ref{table_es_nowVSlater} show normalized value at risk and expected shortfall of $\mathrm{L}_{\bm X} = \widehat{V}_{0}-\widehat{V}_{1}$  and $-\mathrm{L}_{\bm X}$, where $\widehat{V}_t$ stands for either $V_{\bm X, t}$ or $V_{\bm X, t}^{\text{now}}$, for $t=0, 1$. We observe that regress-later outperforms regress-now in all risk measure estimates of long positions, whereas for risk measure estimates of short positions regress-now is best in 4 cases out of 6.\footnote{These observations are in line with the comparison of the detrended Q-Q plots in Figure \ref{fig_regnow}. In fact, Figure \ref{qqplot_regnow_minput} stipulates that the estimation of the right tail of $V_1$ is better with regress-now than with regress-later, which seems in contradiction to the risk measurements of $-\mathrm{L}_{\bm X}$. However, since these risk measures act not only on $\widehat{V}_{1}$, but also depend on $\widehat{V}_{0}$, there is no contradiction with the best risk measurements obtained with regress-later. In fact, if we computed the quantiles of $\widehat{V}_{1}-\widehat{V}_{0}$ and $V_1-V_0$ and plotted the corresponding detrended Q-Q plot, then we would obtain Figure \ref{qqplot_regnow_minput} with a horizontal shift of $-V_{0}$ and a vertical shift of $V_0-\widehat{V}_0$. Our computations suggest that for regress-now and regress-later the vertical shift is downward of same magnitude, which explains that regress-later performs better than regress-now in risk measurements.}
This mixed outcome is somewhat at odds with the above observed superiority of regress-later over regress-now. On the other hand, it serves as an illustration of the no free lunch theorem \cite{wol_mac_97}, which states that there exist no single best method for portfolio valuation and risk management that outperforms all other methods in all situations.

\begin{table}[h]
\centering
  \begin{tabular}{|l|r|r|r|r|}
\hline
Payoff &    $\mathrm{VaR}_{99.5\%}(\mathrm{L})$ & $\mathrm{VaR}_{99.5\%}(\mathrm{L}_{\bm X})$& $\mathrm{VaR}_{99.5\%}(-\mathrm{L})$ & $\mathrm{VaR}_{99.5\%}(-\mathrm{L}_{\bm X})$\\
\hline\hline
Min-put (regress-now) & 2063& 2098& 2058 & 1869\\
Min-put (regress-later) &  2063& \bf{2083}& 2058 & \bf{2123}\\
\hline
Max-call (regress-now) & 2800& 2950& 3071 & \bf{3109}\\
Max-call (regress-later) &  2800& \bf{2802}& 3071 & 2961\\
\hline
Barrier reverse convertible (regress-now) & 264.1& 219.7& 99.83 & \bf{92.50}\\
Barrier reverse convertible (regress-later) &264.1& \bf{264.6}& 99.83 & 85.94\\
\hline
\end{tabular}
\caption{Normalized true and estimated value at risk $\mathrm{VaR}_{99.5\%}(\mathrm{L})/V_0$, $\mathrm{VaR}_{99.5\%}(\mathrm{L}_{\bm X})/V_0$, $\mathrm{VaR}_{99.5\%}(-\mathrm{L})/V_0$, and $\mathrm{VaR}_{99.5\%}(-\mathrm{L}_{\bm X})/V_0$ with $\gamma = 0$. All values are expressed in basis points.}\label{table_var_nowVSlater}
\end{table}

\begin{table}[h]
\centering
  \begin{tabular}{|l|r|r|r|r|}
\hline
Payoff &    $\mathrm{ES}_{99\%}(\mathrm{L})$ & $\mathrm{ES}_{99\%}(\mathrm{L}_{\bm X})$& $\mathrm{ES}_{99\%}(-\mathrm{L})$ & $\mathrm{ES}_{99\%}(-\mathrm{L}_{\bm X})$\\
\hline\hline
Min-put (regress-now) & 2141& 2179& 2118 & 1937\\
Min-put (regress-later) &  2141& \bf{2168} & 2118 &\bf{2219}\\
\hline
Max-call (regress-now) & 2890& 3044& 3205 &\bf{3231}\\
Max-call (regress-later) &  2890& \bf{2880} & 3205 &3090\\
\hline
Barrier reverse convertible (regress-now) & 284.7& 231.8& 101.2 & \bf{92.83}\\
Barrier reverse convertible (regress-later) &  284.7& \bf{283.2} & 101.2 &86.63\\
\hline
\end{tabular}
\caption{Normalized true and estimated expected shortfall $\mathrm{ES}_{99\%}(\mathrm{L})/V_0$, $\mathrm{ES}_{99\%}(\mathrm{L}_{\bm X})/V_0$, $\mathrm{ES}_{99\%}(-\mathrm{L})/V_0$, and $\mathrm{ES}_{99\%}(-\mathrm{L}_{\bm X})/V_0$ with $\gamma = 0$. All values are expressed in basis points.}\label{table_es_nowVSlater}
\end{table}

We note that finite sample guarantees similar to that in Theorem \ref{thmcinew} can be derived for the regress-now estimator, however only under boundedness assumptions on $f$ and $\kappa$. In fact, the boundedness assumptions \eqref{newasstildef} and \eqref{newasstildek} are not enough because they would not guarantee the boundedness of the noise, $f(X)-V_1$. In the literature, boundedness assumption on $f$ is relaxed by assumptions on the noise, see, e.g., \cite{ras_sam_2017}.


\end{appendix}

\newpage

\bibliographystyle{alpha}
\bibliography{Bibliography}

\end{document}